\documentclass[a4paper,UKenglish,cleveref, autoref, thm-restate]{lipics-v2021}

\usepackage[utf8]{inputenc}
\usepackage{subcaption}
\usepackage{graphicx}
\usepackage{colortbl}
\usepackage{float}
\definecolor{light-gray}{gray}{0.95}
\usepackage{mathtools}
\usepackage{etoolbox}

\usepackage{balance}

\usepackage{wrapfig}

\newcommand{\bigO}{\mathcal{O}}

\nolinenumbers

\title{Toward Self-Adjusting $k$-ary Search Tree Networks} 

\author{Evgeniy Feder}{ITMO University, Russia}{dzaba67@yandex.ru}{}{}
\author{Anton Paramonov}{EPFL, Switzerland}{anton.paramonov2000@gmail.com}{}{}
\author{Pavel Mavrin}{Paphos University, Cyprus}{pavel.mavrin@gmail.com}{}{}
\author{Iosif Salem}{TU Berlin, Germany \and ZeroPoint Technologies AB}{iosif.salem@gmail.com}{}{}
\author{Vitaly Aksenov}{City, University of London, UK}{aksenov.vitaly@gmail.com}{}{}
\author{Stefan Schmid}{TU Berlin, Germany}{stefan.schmid@tu-berlin.de}{}{}

\authorrunning{Feder et al.}



\acknowledgements{This paper is partially supported by JetBrains.}

\begin{document}

\maketitle

\begin{abstract}
Datacenter networks are becoming increasingly flexible with the incorporation of new optical communication technologies, such as optical circuit switches, enabling self-adjusting topologies that can adapt to the traffic pattern in a demand-aware manner.
In this paper, we take the first steps toward demand-aware and self-adjusting $k$-ary tree networks. These are more powerful generalizations of existing binary search tree networks (like SplayNet~\cite{splaynet-paper}), which have been at the core of self-adjusting network (SAN) designs. $k$-ary search tree networks are a natural generalization offering nodes of higher degrees, reduced route lengths, and local routing in spite of reconfigurations (due to maintaining the search property). 

Our main results are two online heuristics for self-adjusting $k$-ary tree networks.
Empirical results show that our heuristics work better than SplayNet in most of the real network traces and for average to low locality synthetic traces, and are only a little inferior to SplayNet in all remaining traces. 
We build our online algorithms by first solving the offline case. First, we compute an offline (optimal) static demand-aware network for arbitrary traffic patterns in $\bigO(n^3 \cdot k)$ time via dynamic programming, where $n$ is the number of network nodes (e.g., datacenter racks), and also improve the bound for the special case of uniformly distributed traffic. Then, we present a centroid-based approach to demand-aware network designs that we use both in the offline static and online settings. In the offline uniform-workload case, we construct this centroid network in linear time $\bigO(n)$.

\end{abstract}

\newpage

\section{Introduction}

With more services being offloaded to the cloud 
and the ever increasing numbers of connected devices to the Internet, 
inter- and intra-datacenter traffic grows exponentially.
Therefore, datacenter network design has been attracting a lot of attention.
Traditional datacenter network designs are static and perform well only under certain workloads (e.g. all-to-all).
However, datacenter traffic follows patterns which can be exploited in the design of more efficient networks.
It has been shown that a small fraction of network nodes accounts for a large fraction of datacenter traffic, the traffic distribution is sparse, and it exhibits locality features that change over time \cite{avin2020complexity}.

As a result, innovative dynamic datacenter network topologies have emerged.
Supported by advances in networking hardware (e.g., optical circuit switches, 60 GHz wireless, or even experimental designs based on free-space optics as in \cite{projector-paper}), physical network topologies now have the ability to \textit{self-adjust}. That is, the physical network topology is now programmable and can be reconfigured to serve traffic more efficiently. Interestingly, leading cloud providers have already attempted to incorporate dynamic networks into their datacenters \cite{projector-paper, poutievski2022jupiter}.

This flexibility provided by networking hardware raises an optimization challenge: how to optimally adjust the topology to improve network efficiency in terms of routing cost?
There is a trade-off between the cost of changing the network topology (reconfiguration cost) and the benefit of reducing the distance of frequently communicating racks, henceforth called nodes (routing cost). We assume that the input of this optimization problem is a sequence of communication requests. In the online case, where the input is revealed piecewise and future communication demand is unknown, we would opt for topology updates that are likely to pay off in the future. In the static case, we aim at computing an optimal demand-aware network topology with low time complexity.

The developing field of Self-Adjusting Networks (SANs) aims to address these optimization challenges. SANs often assume a family of allowed topologies, e.g., trees \cite{avin2022deterministic, DBLP:conf/infocom/lazySANsInfocom22, cbnet-paper, peres2021distributed, splaynet-paper}, skip lists \cite{avin2020working}, bounded degree graphs \cite{avin2022demand, avin2021renets}, etc., within which the network has to remain. This restriction is not only practically motivated (e.g., optical switches are of bounded degree), but also simplifies algorithm design and allows for theoretical performance guarantees.
Specifically, self-adjusting tree networks have been at the core of SAN designs.
SplayNet \cite{splaynet-paper}, a self-adjusting binary search tree network generalizing splay trees \cite{splaytree-paper}, was the first proposed SAN. SplayNet has been extended to ReNet \cite{avin2021renets}, a statically optimal SAN for sparse communication patterns, but also to a distributed version, DiSplayNet \cite{peres2021distributed}. 
The search property is particularly useful for SANs, since maintaining it allows local and greedy routing despite changes in the topology.

Online SAN algorithms can vary from being fully reactive, in which case they reconfigure the topology after every communication request~\cite{avin2020working, splaynet-paper} (e.g., when traffic is bursty), to being partially reactive~\cite{DBLP:conf/infocom/lazySANsInfocom22}, in which case they update the topology periodically. 
In \cite{DBLP:conf/infocom/lazySANsInfocom22}, the topology changes every time the routing cost reaches a threshold $\alpha$ since the last topology update, the new topology is computed using SplayNet, and it remains static until the routing cost reaches the threshold again.
This approach can be generalized to a meta-algorithm, where the topology changes upon a new chunk of the input, a subroutine is used to decide the new topology, and in between reconfigurations it remains static. 
Therefore, the efficient computation of static \textit{demand-aware} topologies is also relevant in online SAN algorithm design.   



In this work, we take the first steps to generalizing binary to $k$-ary search tree networks, 
since they provide higher node degrees and shorter routes than binary search trees (BSTs) for a fixed number of nodes, in addition to local and greedy routing regardless of reconfigurations due to the search property. 
We present offline static and online self-adjusting networks and evaluate our newly proposed SANs experimentally. 

We note that designing $k$-ary search tree networks is different than designing a $k$-ary search tree. This holds due to the need of having a node identifier that stays the same across rotations in the network case 
(the assignment of identifiers to nodes is a bijection and does not change)
as each tree node represents a network node.
That is, in the case of $k$-ary search trees each node contains up to $k - 1$ keys which are used for routing (traversing from root to searched key) in the data structure and are also the data.
In contrast, in the network case each node can use up to $k - 1$ nodes' keys for routing (hence called routing keys), but each node should have an extra fixed key that serves as its identifier (node key) which is the data.
Self-adjusting
$k$-ary search trees have been studied, for example, by Sherk in \cite{sherk1995self} and by Martel in \cite{martel1991self}.
However, none of the existing approaches, to the best of our knowledge, apply in our case due to the requirement of having one key (identifier) per node.

\noindent\textbf{Contributions.}
We present offline static and online self-adjusting $k$-ary search tree networks. 
As for the static results, we have:

1. We construct a static ``almost-optimal'' $k$-ary tree network in $O(n^3 \cdot k)$ time using dynamic programming. Then, we reduce the complexity to $O(n^2 \cdot k)$ for the special case of uniformly distributed traffic to obtain optimal tree. The latter case is non-trivial, since we do not restrict the topology to full balanced trees.

2. We present a new static $k$-ary tree network topology which is built by $k+1$ trees with almost equal size connected around a centroid node. 
We present a linear offline algorithm that constructs it and prove that its total cost is close to the optimal.

As for the dynamic case, we have:

1. We then present two novel variants for online $k$-ary self-adjusting search tree networks: $k$-ary SplayNet, which is a generalization of SplayNet, and $(k+1)$-SplayNet which is obtained by applying a centroid heuristic. For both networks, we propose novel rotation operations that allow network nodes to keep their identifiers across tree rotations.
To achieve the latter, we distinguish the set of node identifiers from the set of routing keys.
We show how the known upper bounds for SplayNet can be applied for our networks.

2. We perform two types of experiments with synthetic network traces and real ones from datacenter network traffic. First, we show that our $k$-ary SplayNet indeed has better routing cost than SplayNet. Then, we compare 3-SplayNet to SplayNet, and to static demand-aware and demand-oblivious trees. The results show that 3-SplayNet (centroid topology) performs better than the standard SplayNet in most real network traces and in synthetic ones that have average to low temporal locality, while its performance is similar to SplayNet for the remaining datasets.

\textbf{Related work. } 
Self-Adjusting Networks were introduced with SplayNet \cite{splaynet-paper}. SplayNet is a binary search tree network that generalizes Splay Trees \cite{splaytree-paper}. SplayNet uses the tree rotations of splay trees to reduce the distance of communicating nodes to one; it uses splay operations to move the source and destination to their lowest common ancestor for each communication request. The same paper presents a dynamic programming algorithm for computing an optimal binary search tree network when the demand is known, among other results. SANs were further surveyed and classified in \cite{avin2019toward}. The authors present a classification of network topologies, which depends on whether they are (i) oblivious to or aware of traffic patterns, (ii) fixed or reconfigurable, and (iii) aware of the input sequence of communication requests (offline, online, generated by a distribution). This taxonomy allows for optimizing for certain properties according to each case, e.g. diameter or competitive ratio. A survey including first solutions and enabling networking technologies can be found in \cite{hall2021survey}.

Tree-based SANs were further studied following SplayNet, due to being more easy to analyze and deploy. 
ReNets \cite{avin2021renets} are bounded-degree SANs based on combining ego-trees, which are trees where the source is a node and the remaining nodes are the destinations to which the source node has communicated with. In this design, ego-trees are stars or splay trees, depending on whether the number of destinations exceeds the degree bound. ReNets achieve static optimality for sparse communication patterns, which is a desirable optimality property \cite{avin2013locally, avin2019toward}.
Ego-trees were further studied in the form of self-adjusting single-source tree networks in \cite{avin2022deterministic, avin2022push}, which provided a number of constant competitive (dynamically optimal) randomized and deterministic algorithms with good experimental performance.
SplayNet has also been the basis of distributed tree SANs \cite{cbnet-paper, peres2021distributed} and of SANs in a cost model with non-unit cost for changing a link in the topology \cite{DBLP:conf/infocom/lazySANsInfocom22}.  
All the results mentioned above are for binary tree networks.

Alternative directions have also been studied.
\cite{avin2019demand} studied how to construct offline SANs when the demand is known, under certain assumptions on the communication patterns (e.g. sparse demand). 
\cite{griner2021cerberus} presents a topology adjustment algorithm that uses static or dynamic topologies according to the identified traffic patterns (latency-sensitive, all-to-all, elephant flows).
SANs that are not tree-based  have also been studied, e.g. Skip List Networks \cite{avin2020working}.


\textbf{Paper organization. } 
In Section~\ref{sec:model}, we introduce all the necessary definitions.
In Section~\ref{sec:static}, we explain how to build demand-aware optimal tree network using dynamic programming and how to build a quasi-optimal tree for the uniform workload.
In Section~\ref{sec:online}, we present novel rotations for $k$-ary SplayNet and present two heuristics.
In Section~\ref{sec:experiments}, we experimentally evaluate the cost of our new network structures.
Finally, we conclude in Section~\ref{sec:conclusion}.


\section{Model}
\label{sec:model}

We consider a network of $n$ nodes $V = \{1, \ldots, n\}$ (e.g., top-of-the-rack switches in a datacenter networks) and a finite or infinite communication sequence $\sigma = (\sigma_1, \sigma_2, \ldots)$, where  $\sigma_t = (u, v) \in V^2$ is a communication request from source $u$ to destination $v$.
%
The network topology $G$ must be chosen from a family of desired topologies $\mathcal{G}$, for example, search trees, expander graphs, etc. Each topology $G \in \mathcal{G}$ is a graph $G = (V, E)$. 
The routing cost of $\sigma_t$ is given by the distance between the two endpoints in the topology when serving the request.
The topology can be reconfigured between requests with a cost equal to the number of links (edges) added or removed.
The total service cost of $\sigma$ is the sum of routing and reconfiguration costs.
Our goal is to serve the communication sequence with minimum total cost.



We distinguish two problem variants. In the \textbf{offline static} variant, $\sigma$ is known in the form of an $n\times n$ demand matrix $D$, but no reconfiguration can occur. The matrix entry $D[u,v]$ is the number of requests from $u$ to $v$ in $\sigma$. 
We have to build a network topology $G_{static} \in \mathcal{G}$ that does not change during or in between requests. Such a graph $G_{static}$ needs to optimize the total distance function,
	\begin{align*}
		\mathrm{TotalDistance}(D, G_{static}) = \sum\limits_{(u, v) \in [n]^2}d_{G_{static}}(u, v)\cdot D[u, v]
	\end{align*} 
	where $d_{G_{static}}(u, v)$ is the distance between nodes $u$, $v$ in $G_{static}$ and $[n] = \{1,\ldots, n\}$.

In the \textbf{online self-adjusting} variant, $\sigma$ is not known in advance but revealed piecewise, and we can change the topology after serving a request. We are provided with an arbitrary initial network (before the first request arrives), which we denote by $G_0 \in \mathcal{G}$. Our task is to build an online algorithm $\mathcal{A}$ that adjusts the network $G_i$ at every time instant $i=1,\ldots,m$ and minimizes the total cost, which is calculated as $\mathrm{sumCost}(\mathcal{A}, G_0, \sigma) = \sum_{i=1}^{m} \left(\mathrm{routingCost}(G_{i-1}, \sigma_i)\right.$ $+$ $\mathrm{adjustmentCost}(G_{i-1}$, $\left. G_i)\right)$, where $\mathrm{routingCost}(G_{i-1}, \sigma_i)$ is the path length in edges of $G_{i-1}$ to route request $\sigma_i$ and $\mathrm{adjustmentCost}(G_{i-1}, G_i)$ is the adjustment cost to reconfigure the network from step $i-1$, $G_{i-1} \in \mathcal{G}$, to step $i$, $G_i \in \mathcal{G}$, i.e., the number of edges added or deleted. Note that $\mathcal{A}$ can skip a reconfiguration step.


This paper focuses on both problem variants when the set of allowed topologies $\mathcal{G}$ is the set of $k$-ary search trees. These trees are the generalization of binary search trees, which were investigated in~\cite{splaynet-paper} in the context of SANs. The main advantage of using search trees as self-adjusting network topologies is that we can  route locally and greedily: given a destination identifier (or address), each node can decide locally to which neighbor to forward the packet using the search property.
This is particularly useful in the online setting, as routing tables do not need to be updated upon reconfiguration: a node given a packet can just use the information from the routing keys to forward the packet, accordingly.
Also, with increasing $k$, route lengths decrease and node degrees increase.

\begin{definition}
\label{def:karysearchtree}
(i) A \emph{$k$-ary Search Tree} is a rooted tree on keys (node identifiers) $1, \ldots, n$, where each node stores a key (node identifier), a routing array $r=(r_1,r_2,\ldots, r_{k-1})$ containing routing elements (not keys), and has at most $k$ children defined by $r$ as follows: keys of nodes in the $i$-th child are between $r_i$ and $r_{i+1}$ for $i \in [1, k-2]$, keys of nodes to the left of $r_1$ are smaller than $r_1$, and keys of nodes to the right of $r_{k-1}$ are larger than $r_{k-1}$.
Note that the key does not necessarily belong in the routing array.

(ii) A \emph{routing-based} $k$-ary Search Tree is a $k$-ary search tree in which the node identifiers are contained in the routing array.
\end{definition}

\begin{wrapfigure}{r}{0.5\textwidth}
    \centering
    \includegraphics[width=\linewidth]{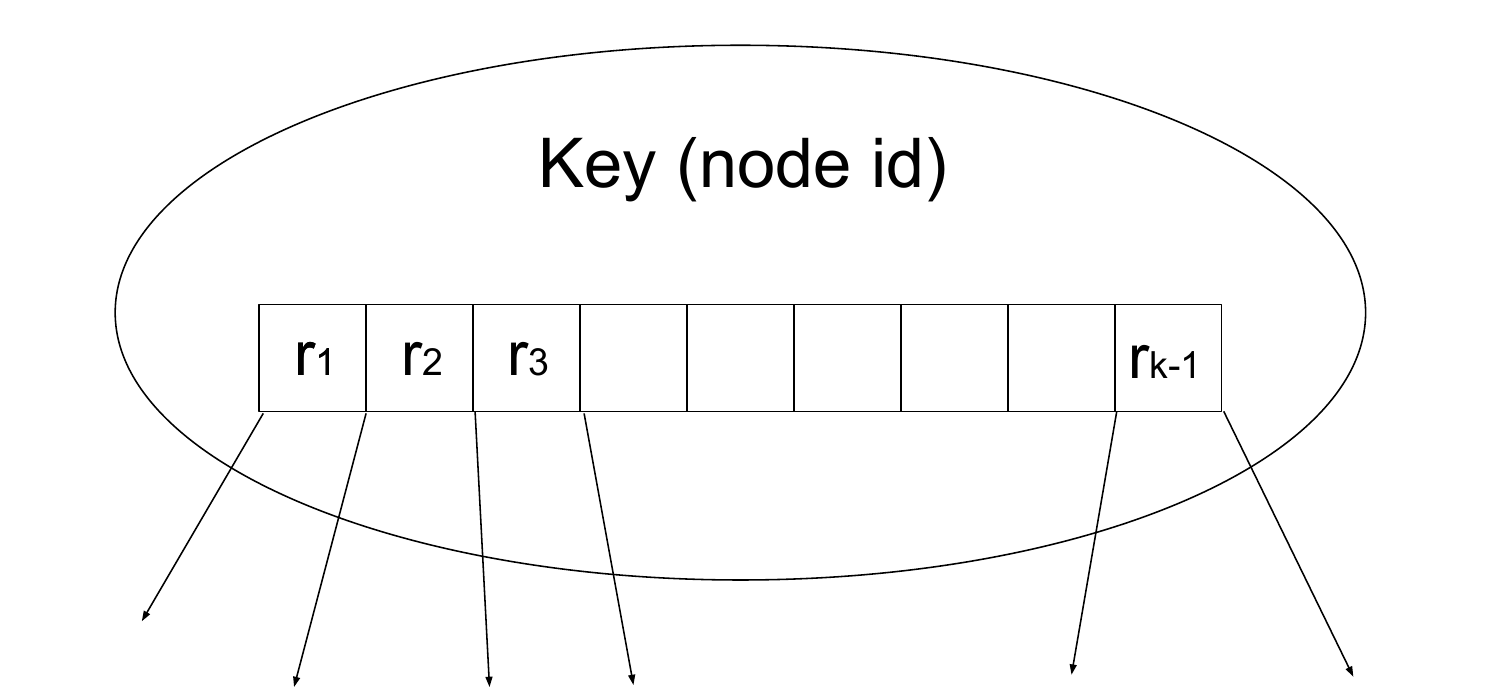}
    \caption{Node in $k$-ary search tree}
\end{wrapfigure}

Definition \ref{def:karysearchtree} defines the standard search route in a $k$-ary tree, when starting from the root. A routing path between two nodes is the unique path connecting them. Routing from a node containing source routing key $i$ to a destination node containing routing key $j$ occurs by following the upward (reverse search) path until their lowest common ancestor and then the standard downward search path to the destination. 
The local transformations of search tree networks are called rotations. A rotation in a $k$-ary search tree changes some adjacency relationships, while keeping subtrees intact and maintaining the search property (Definition \ref{def:karysearchtree}). Such rotations can be implemented in a different manner. We introduce $k$-splay, a novel rotation procedure, in Section~\ref{sec:ksplaynet}.

\section{Optimal static $k$-ary search tree networks}
\label{sec:static}

In this section we construct an optimal static \emph{routing-based} $k$-ary search tree network via dynamic programming.
By analysing the optimal tree for a uniform workload, we show a linear-time but non-exact offline static algorithm that constructs a topology with a centroid node and $k+1$ trees connected to it for the chosen workload.
The latter result will be the basis for an online self-adjusting network proposed in Section \ref{sec:ksplaynet}.

\subsection{Dynamic programming algorithms}
\label{sec:dp}
As our first result, we construct an offline (optimal) static routing-based $k$-ary search tree network. The algorithm is similar to the one for binary search tree from~\cite{splaynet-paper}: it just works as the dynamic programming on segments but a little bit more involved. Its complexity is $O(n^3 \cdot k)$.

\begin{theorem}
\label{thm:arbitrary workload}
An offline static routing-based $k$-ary Search Tree network, i.e., one with the minimal total distance given the requests in advance, can be constructed in $O(n^3 \cdot k)$.
\end{theorem}

\begin{remark}
  We remark that computing an optimal static non-routing-based $k$-ary search tree network is an open problem. That is, there is neither a trivial dynamic programming extension for it nor an NP-hardness proof.
\end{remark}

Then, we decided to consider a uniform workload. That workload is an infinite workload where each pair of nodes is requested uniformly at random.
In this case, we can improve the algorithm to $O(n^2 \cdot k)$ since the dynamic programming in this case does not depend on the position of the segment and depends only on its length.
Our goal is to find a static $k$-ary search tree that serves an infinite uniform workload as fast as possible, i.e., the expectation of the cost of each query is minimal.
So, we can represent the uniform workload as a workload with the demand matrix filled with ones.

\begin{theorem}
\label{thm:uniform workload}
An offline static $k$-ary Search Tree network for the uniform workload can be constructed in $O(n^2 \cdot k)$.
\end{theorem}

We remark that the resulting tree for the uniform workload is not required to be \emph{routing-based} as in the generic case.

More details on these algorithms appear in Appendix~\ref{app:dp}.

\subsection{Centroid static $k$-ary search tree network in $O(n)$ time on uniform workload}
\label{sec:quasiopt}

Now, we present a linear-time construction of an almost-optimal offline $k$-ary search tree network for the uniform workload. We show the cost difference to the optimal tree in Theorem \ref{thm:quasiopt} and comment on experimental results in Remark \ref{rem:expcentroid}. 
This construction will serve us as a basis of an online heuristic in Section~\ref{sec:ksplaynet}.
In general, faster computations of static network topologies are relevant for scaling to larger self-adjusting networks for which we compute new (demand-aware) topologies periodically.
Also, the uniform workload is relevant to the all-to-all traffic pattern.

A $k$-ary search tree can be split in levels: the $i$-th level consists of nodes that are at distance $i - 1$ from the root.
The tree is \emph{weakly-complete} when all its levels, except for the last one, are fully filled (i.e., the $i$-th level has $k^{i-1}$ nodes).
Nodes on the last level can be distributed arbitrarily as long as the search property holds.
%
%

Since in this case we consider only the uniform workload we can ignore the search property: we can first fix the tree structure and then distribute the keys so that the search property is respected.
Our goal is to find an optimal $(k+1)$-degree tree (instead of a $k$-ary search tree) that minimizes $\mathrm{TotalDistance}(D_{uniform}, T) = \sum\limits_{(u, v) \in [n]\times[n]}d_{T}(u, v)$. A $(k+1)$-degree tree is a non-rooted tree where each node has at most $k+1$ neighbours. Such trees represent the same set of trees as $k$-ary search trees: you can root a $(k+1)$-degree tree by a leaf and obtain a $k$-ary search tree (by correctly distributing keys).
Intuitively, we propose a topology where the root has $k+1$ children, in contrast to common a $k$-ary tree with a $k$-degree root, with the potential of reducing the total cost of routing communication requests.

\begin{definition}
A \emph{centroid} $(k+1)$-degree tree is a tree with the root having $k + 1$ weakly-complete $k$-ary trees. All the levels of the tree are fully filled, possibly except for the last one. We can change the relative positions of subtrees such that the leaves on the last level are all grouped together to the left. The tree is shown on Figure~\ref{fig:quasi-optimal-tree-main}.
\end{definition}

\begin{wrapfigure}{r}{0.5\textwidth}
    \centering
    \includegraphics[width=\linewidth]{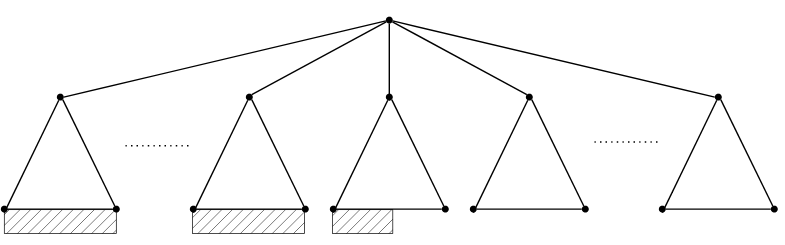}
    \caption{A centroid tree after the reposition of subtrees. Lined rectangles represent leaves.}
    \label{fig:quasi-optimal-tree-main}
\end{wrapfigure}




We prove (Theorem \ref{thm:quasiopt}) that the centroid $(k+1)$-degree tree $T$ has total distance in the uniform workload close to the total distance in the optimal $(k+1)$-degree tree.
Intuitively, the proof works as follows. First, we show that all sibling subtrees, i.e., subtrees with the same parent, should either have the same height or their heights should be different by one. If not, we can move the leaf from one subtree to another and decrease the total cost. The latter move does not always decrease the cost, but sometimes it leads to a small overhead. Then, we show that if for some subtree not all leaves on the last level are aligned to the ``left'', then we can move them and decrease the cost. By that we show that the centroid tree is almost optimal.
The complete proof is provided in Appendix~\ref{app:quasiopt}.
Here we just state the main results.

\begin{theorem}
\label{thm:quasiopt}
Assuming $k$ is a constant, the difference in the total distance between an optimal $(k+1)$-degree tree $T$ and our centroid $(k+1)$-degree tree is $O(n^2k\log k)$ while the total distance in the optimal $(k+1)$-degree tree is $\Omega(n^2\log n)$.
\end{theorem}


\begin{remark}
We can get a $k$-ary search tree out of $(k+1)$-degree centroid tree by rooting at some leaf and setting the identifiers correspondingly. We name such a tree a \emph{centroid $k$-ary search tree}.

Since we consider the uniform workload we know that our centroid tree has the total cost of requests close to the optimal, i.e., misses by at most $O(n^2)$ while the total optimal cost is $\Omega(n^2 \log n)$. Thus, our centroid $k$-ary search tree has an approximation ratio $1+O(\frac{1}{\log n})$.
\end{remark}

Then, it is quite straightforward to build that centroid tree.
\begin{theorem}
The centroid $k$-ary search tree can be built in $O(n)$ time.
\end{theorem}

Finally, we show that the full $k$-ary tree also has total distance close to the cost of the optimal tree.
\begin{lemma}
The total distance in the full $k$-ary tree and the total distance in the centroid $(k+1)$-degree tree are both $n^2\log_kn + O(n^2)$. That is, their total distance differs from the total distance in the optimal tree by $O(n^2)$.
\end{lemma}

\begin{remark}
\label{rem:expcentroid}
The results of the last Lemma show that the full and centroid trees are close to the optimal. However, in the uniform workload the centroid tree should have better total cost, since we split in the centroid vertex by $k+1$ balanced subtrees.
In our experiments, we found that our \emph{centroid} $k$-ary search tree \textbf{is indeed optimal} for all $n$ less than $10^3$ when $k$ is up to $10$, but we were not able prove its optimality formally.
\end{remark}

\section{Online self-adjusting $k$-ary search tree networks}
\label{sec:online}


\label{sec:ksplaynet}



We present online algorithms for self-adjusting $k$-ary search tree networks. 
%
The first one is the $k$-ary SplayNet, which is a self-adjusting network based on a $k$-ary search tree and a generalization of SplayNet~\cite{splaynet-paper}. We prove that $k$-ary SplayNet has the same complexity bounds as SplayNet (we show its benefits experimentally in Section \ref{sec:experiments}).   
The second one is $(k+1)$-SplayNet, a centroid-based structure, that is based on $k$-ary SplayNet and the centroid $k$-ary search tree network presented in the previous section.

\subsection{$k$-ary SplayNet}

In the literature, only one proposal for $k$-ary self-adjusting trees exists~\cite{sherk1995self}. The tree rotations proposed in~\cite{sherk1995self} cannot be directly generalized to SANs: multiple keys appear in each node and a node's keys change upon a rotation, so they cannot be used as network node identifiers and it is not clear how to maintain node identifiers across tree rotations.
%
In this section, we propose new splay operations: $k$-semi-splay and $k$-splay. These rotations mimic the rotations in the binary splay tree and allow for persistent node identifiers, while re-shuffling routing arrays.
For example, the routing array of the node with identifier $X$ in Figure \ref{fig:semi-splay-1} is $(a_1, a_2, a_3, \ldots, a_{k-1})$ and key $X$ has value in between $a_3$ and $a_4$.
Thus, any key can be located by the search property and the routing algorithm is the same as in SplayNet \cite{splaynet-paper}: we rotate the source and destination to the lowest common ancestor and then route the request through a direct link. We now present splay operations that preserve the search property.


\vspace{0.1cm}
\begin{minipage}{0.5\textwidth}
    \centering
    \includegraphics[height=6cm]{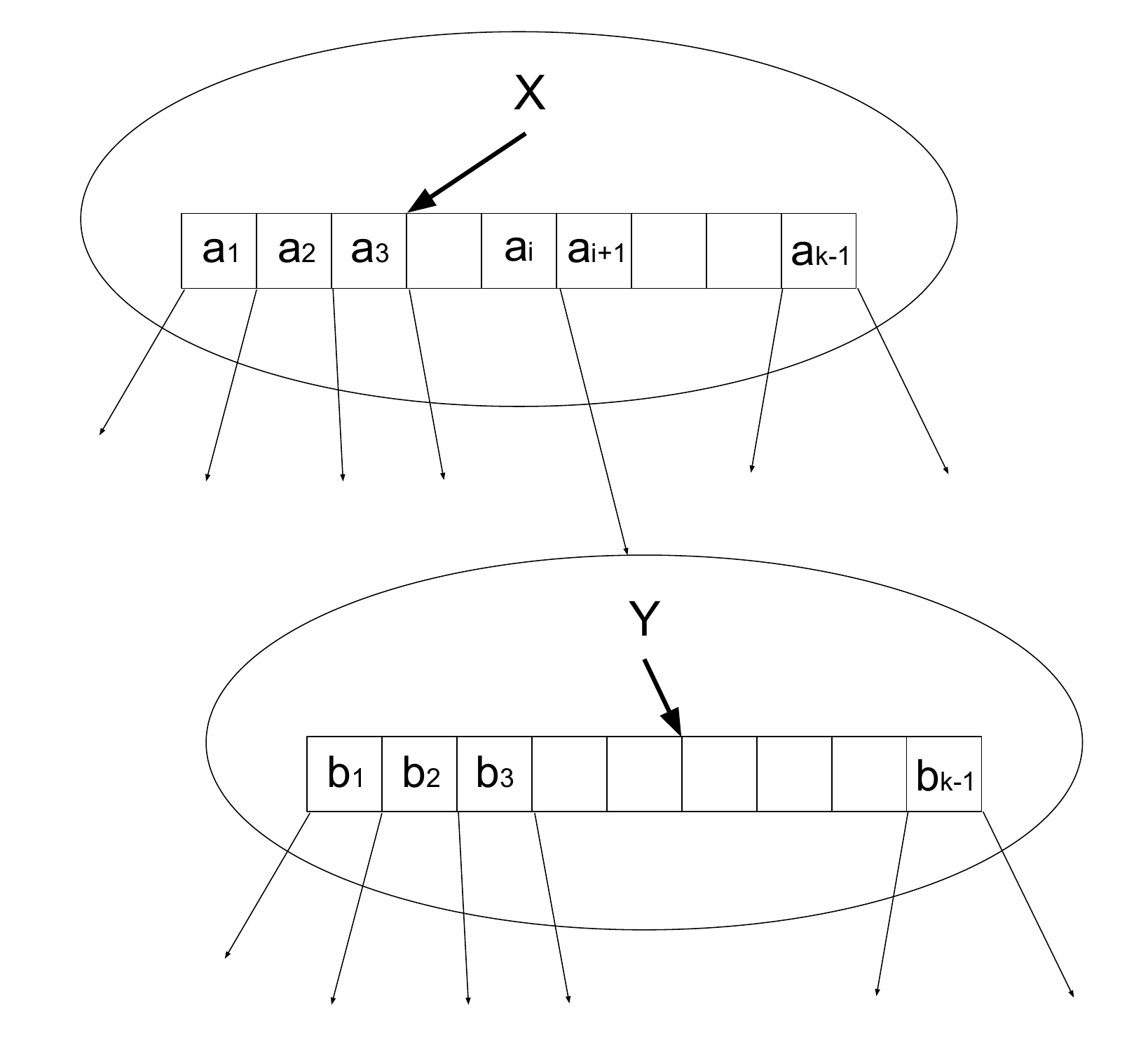}
    \captionof{figure}{The initial state for $k$-semi-splay.}
    \label{fig:semi-splay-1}
\end{minipage}%
\begin{minipage}{0.5\textwidth}
    \centering
    \includegraphics[height=6cm]{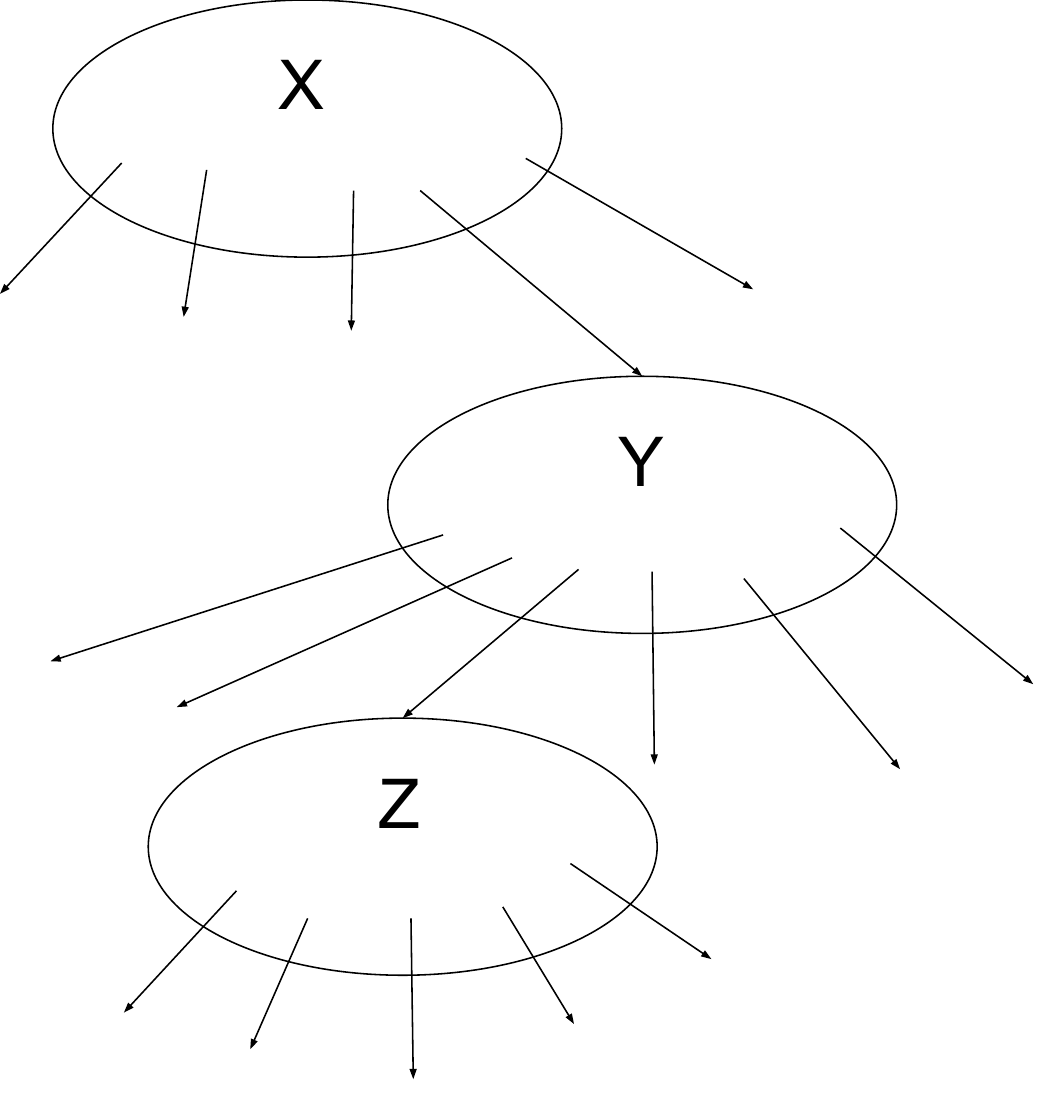}
    \captionof{figure}{Example state before $k$-splay.}
    \label{fig:splay-1}
\end{minipage}
\vspace{0.1cm}


We first present $k$-semi-splay, which generalizes the zig and zag operations in splay trees \cite{splaytree-paper}. 
Suppose that we have two nodes: a node with id $X$ and a child node with id $Y$ (Figure~\ref{fig:semi-splay-1}).
Our goal is to make $Y$ a parent of $X$.
To that end, we merge the routing arrays from these nodes and search for the position for $X$ in this array.
Then, we take $X$ as a key and some $k-1$ consecutive routing elements ``covering'' $X$ as a routing array, i.e., $X$ as the value lies in the segment built on a consecutive set of $k - 1$ elements. 
Finally, we set this node as a child of a new node with key $Y$ and the routing elements left.

Now, we explain the $k$-splay rotation, which generalizes the combination of two zig or zag operations in splay trees. Suppose we have three nodes with identifiers $X$, $Y$, and $Z$ (Figure~\ref{fig:splay-1}), and we want to make $Z$ the top node.
First, we merge the routing arrays of these three nodes into one array and find the positions of $X$ and $Y$ there.
There are two cases: 1)~$X$ and $Y$ are located distant to each other, i.e., separated by more than $k-1$ routing elements; or 2)~$X$ and $Y$ are close to each other.
In the first case, we make two new nodes: one with $X$ and $k-1$ consecutive routing elements ``covering'' $X$; and the one similar for $Y$.
Finally, we set these new nodes as children of a node with key $Z$ and routing elements left (Figure~\ref{fig:splay-2}).
In the second case, we make two new nodes: the one with $Y$ is a parent of the new node with $X$.
Then, we set the node with $Y$ as a child of a node with key $Z$ and routing elements left (Figure~\ref{fig:splay-3}).

\begin{figure}
    \begin{minipage}{0.5\textwidth}
    \centering
    \includegraphics[width=\linewidth]{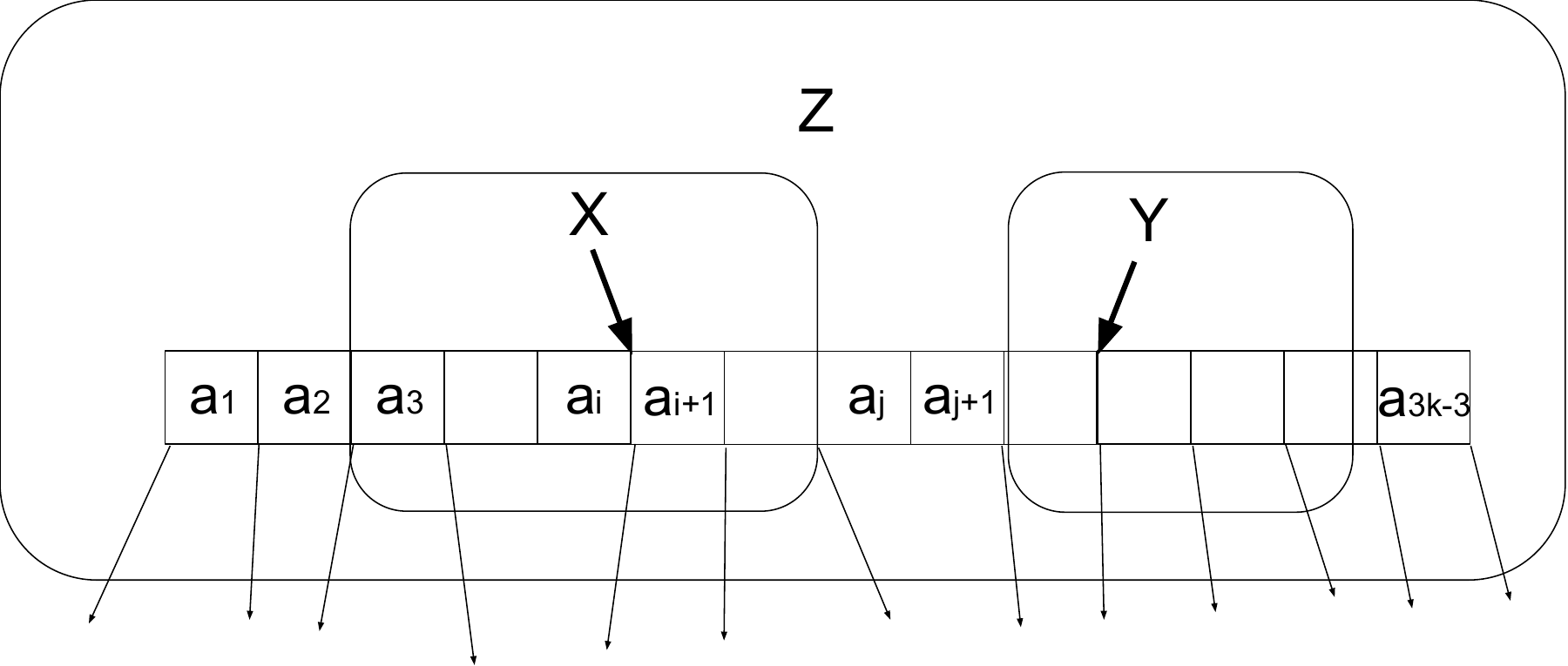}
    \caption{$k$-splay rotation. The first case.}
    \label{fig:splay-2}
    \end{minipage}
    \begin{minipage}{0.5\textwidth}
    \centering
    \includegraphics[width=\linewidth]{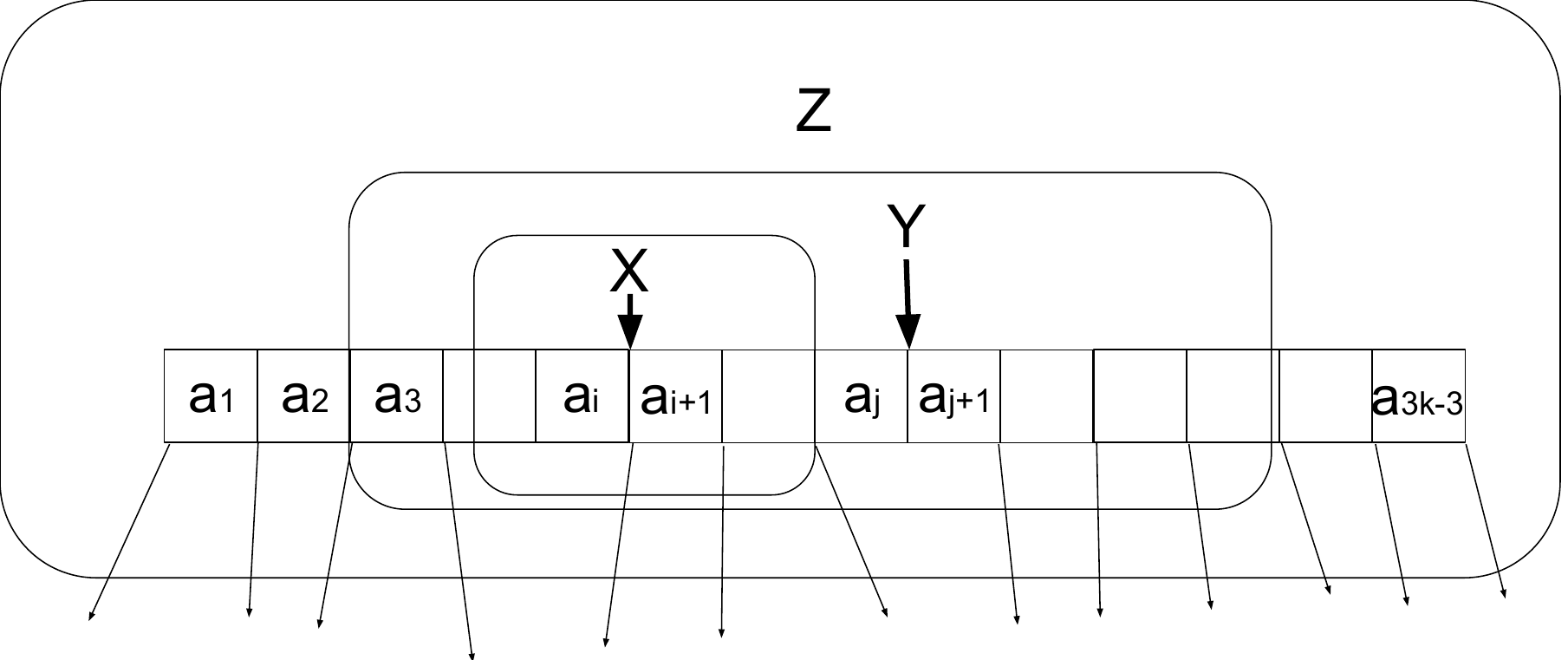}
    \captionof{figure}{$k$-splay rotation. The second case.}
    \label{fig:splay-3}
    \end{minipage}
\end{figure}

Using these rotations, we design a data structure called $k$-ary SplayNet which works similarly to binary SplayNet: upon serving a request between two keys, we use $k$-splay and $k$-semi-splay rotations to move the nodes to their lowest common ancestor, replacing it and one of its children.
By that, after adjustments, the request can be served in constant time.

Here, we designed only two types of rotations. These rotations can be seen as the generalization of the zig/zag and zigzag rotations in binary splay tree.
Now, we prove that a $k$-ary splay tree, i.e., all the routing requests are from the root, based on our rotations is statically-optimal. Please note that the lower bound to serve the search requests for $k$-ary splay tree is the same as for the standard splay tree~\cite{splaytree-paper}: suppose the lower bound is asymptotically better for $k$-ary splay tree, then we can represent $k$-ary tree as the splay tree with the constant multiplier $\log k$ on the total cost.

\begin{remark}
    The $k$-semi-splay and $k$-splay operations cannot be applied to routing-based trees. For example, after collapsing three nodes in $k$-splay (the second case, as in Figure~\ref{fig:splay-3}) we might get routing keys $Z$, $Y$, and $X$ (in this order) as the leftmost ones in the joint array of routing keys. At that point, we cannot distribute the routing keys in the required order, since these keys $X$, $Y$, and $Z$ should be node identifiers and routing keys at the same time and $X$ should be on top. When taking $X$ as a node identifier and a routing key of the topmost node, we will get only $Y$ and $Z$ keys in the leftmost subtree, which restricts us from filling in $X$ with routing keys.
\end{remark}

\begin{theorem}
The $k$-ary splay tree based on the $k$-semi-splay and $k$-splay rotations is statically-optimal. In other words, the total cost to serve $n_x$ search requests to node $x$ is $O(m + \sum_x n_x \cdot \log \frac{m}{n_x})$ where $m$ is the total number of requests.
\end{theorem}

\begin{proof}[Sketch]
The theorem can be proven in the almost identical way as the result for binary splay tree~\cite{splaytree-paper}. As in the splay tree proof, we present the potential of the node with key $v$ as $r(v) = \log w(v)$, where $w(v)$ is the total weight of all nodes in the subtree of $v$.
Then, we can reprove the Access Lemma from~\cite{splaytree-paper}, where the difference of the potential moving a node to the top does not exceed $3 \cdot (r(\texttt{root}) - r(v)) + 1$.
This lemma holds since all the inequalities on the potentials of nodes from the proof in~\cite{splaytree-paper} remain the same since $k$-semi-splay changes the potential exactly as zig does; our first case rotation of $k$-splay changes the potential exactly as zig-zag does, and, finally, our second case rotation of $k$-splay changes the potential exactly as like zig-zig does.
Thus, with the proved access lemma we can get the main theorem.
\end{proof}

Because of this theorem, the complexity bound proven for binary SplayNet in~\cite{splaynet-paper} holds for our $k$-ary SplayNet.

\begin{theorem}
$k$-ary SplayNet performs requests $\sigma = ((u_1, v_1), (u_2, v_2), \ldots, (u_m, v_m))$ with the cost of entropies of sources and destination: $O(\sum\limits_{x=1}^{n} a_x \cdot \log \frac{m}{a_x} + b_x \cdot \log \frac{m}{b_x})$, where $a_x$ is the number of requests with $x$ as a source and $b_x$ is the number of requests with $x$ as a destination.
\end{theorem}

Moreover, since each node has $k$ children instead of two, we expect that the total routing costs in $k$-ary SplayNet are smaller than in the original SplayNet. The experimental comparison between them appears in Section~\ref{sec:karysplaynet}.

We note that there can be more alternatives to $k$-semi-splay and $k$-splay. For example, we can take any $d$ connected nodes in the tree and modify them in a manner that the node with a chosen key will be in the topmost one after the update. This can be done as follows: 1)~merge all $d$ routing arrays into one; 2) find the positions of our $d$ identifiers in this array; 3)~choose some order of keys $k_1, k_2, \ldots, k_d$ in the nodes; 4) consider the $i$-th key $k_i$, take the $k-1$ consecutive routing keys ``covering'' $k_i$, and use them to form a new node with key $k_i$; 5) remove these routing elements from the total routing array and repeat the previous phase for next keys. At the end, the topmost node will contain the required key $k_d$.
Thus, we can have different versions of $k$-ary SplayNet depending on the rotations we choose.
In the remainder, we refer to $k$-ary SplayNet as any (black-box) implementation that maintains the search property and one identifier per node.

\subsection{Application of the centroid heuristic}

Our theoretical studies in Section~\ref{sec:quasiopt} show that the total cost for the uniform workload of both structures, the full $k$-ary tree and the centroid one, is very close. But we know from our experiments that for $n$ up to $10^3$ and $k$ up to 10 the centroid tree is actually optimal (cf. Remark~\ref{rem:expOptCentroid} in the Appendix).
With this practical motivation, we designed an online heuristic based on the centroid idea. We present $(k+1)$-SplayNet, which is a centroid-based structure and the online self-adjusting equivalent of the static tree from Section \ref{sec:quasiopt}.
The topology is presented in Figure~\ref{fig:k-splaynet}.
We split the nodes in $k+1$ almost equal parts and specify two centroid nodes: $c_1$ and $c_2$.
Centroid $c_2$ corresponds to the centroid from the previous section; its subtrees have $(n - 2) / (k + 1)$ nodes.

\vspace{0.2cm}
\noindent
\begin{minipage}{.4\textwidth}
\centering
\includegraphics[height=3.4cm]{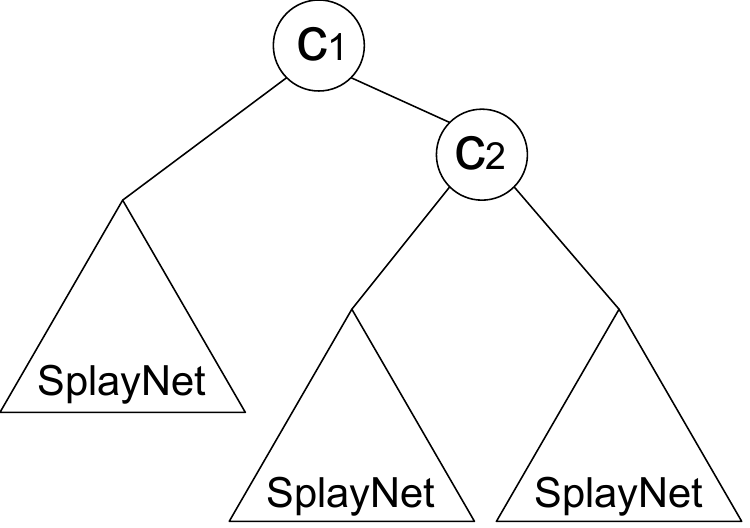}
\captionof{figure}{3-SplayNet structure}
\label{fig:3splaynet}
\end{minipage}%
\begin{minipage}{.6\textwidth}
\begin{center}
\includegraphics[height=3.4cm]{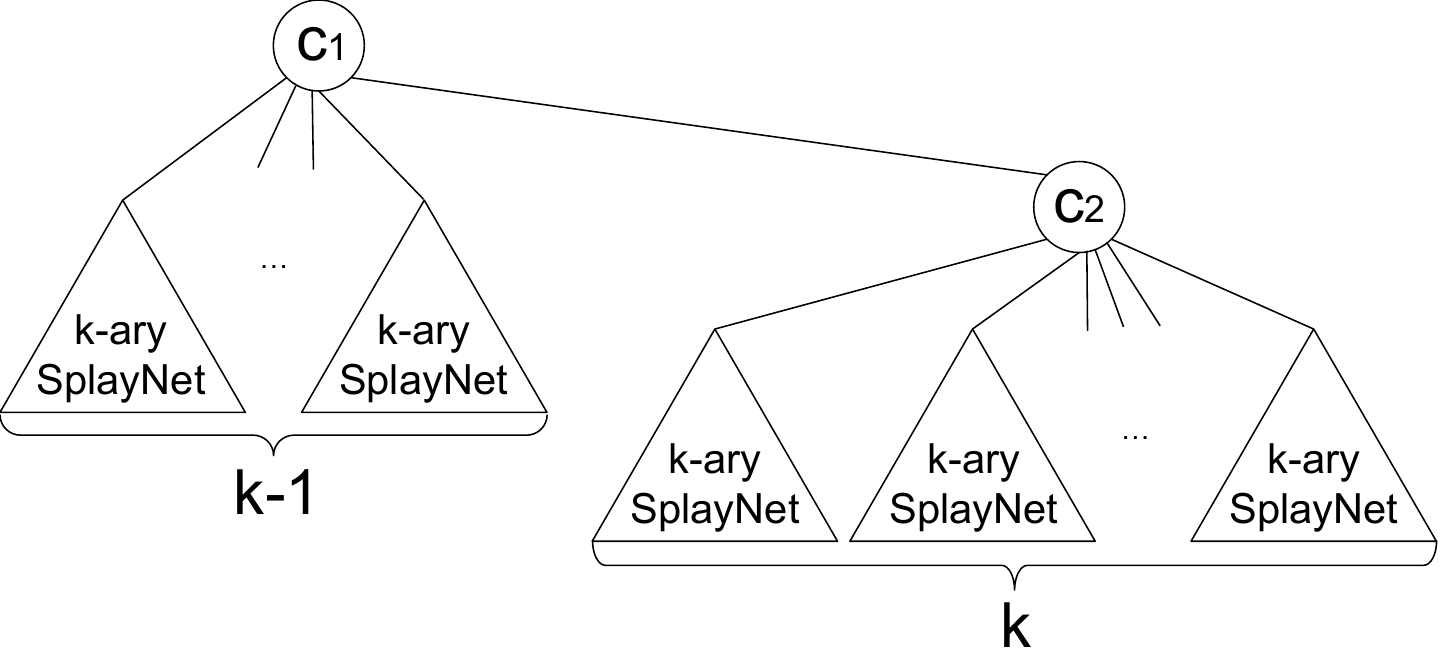}
\captionof{figure}{$(k+1)$-SplayNet structure.}
\label{fig:k-splaynet}
\end{center}
\end{minipage}
\vspace{0.2cm}

Centroid $c_1$ has $k-1$ children (in addition to $c_2$), that are $k$-ary SplayNets of size $[(n - 2) / (k + 1)] / (k - 1)$. 
The $k$ children of $c_2$ are $k$-ary SplayNets of size $(n-2) / (k + 1)$. When serving a request $(u, v)$, we $k$-splay $u$ and $v$ to their lowest common ancestor, as was done in $k$-ary SplayNet, but we never move nodes $c_1$ and $c_2$.
That is, requests within the same subtree are served exactly as in $k$-ary SplayNet and for requests originating in different subtrees of $c_1$ and $c_2$, we splay the endpoints to their subtree roots and then we route the request via the path $u\to c_1 \to c_2 \to v$.
The sets of nodes in the $2k-1$ subtrees remain intact, but these subtrees still can self-adjust.
We study $(k+1)$-SplayNet and $k$-ary SplayNet experimentally in the next section.

\section{Experimental evaluation}
\label{sec:experiments}
We have two presented approaches to evaluate: 1)~we compare the cost of our $k$-ary SplayNet with the static balanced $k$-ary search tree, the optimal static routing-based $k$-ary search tree (Section \ref{sec:dp}), and the standard SplayNet, i.e., $2$-ary SplayNet; 2)~we compare our new 3-SplayNet based on the centroid heuristic ($(k+1)$-SplayNet for $k=2$) with the standard SplayNet and the two static data structures from item 1), for $k=2$.
That is, in the first item we study the benefit of increasing $k$ across diverse workloads, while in the second item we compare the two heuristics of Section \ref{sec:online} and two static trees for $k=2$. In all our experiments, we set the routing and rotation costs to one.

\paragraph*{Setup and data} The code for the algorithms was written in C++ and Python. We perform three types of experiments: (i) on the uniform workload with $100$ nodes, (ii) on synthetic workloads with $1023$ nodes and the temporal complexity parameter (the probability of repeating the last request~\cite{avin2020complexity}) taking the values  $0.25$, $0.5$, $0.75$, and $0.9$, and (iii) on the data from three real-world datasets of datacenter network traces: a high performance computing (HPC) workload~\cite{doe2016characterization}, a workload on ProjectToR~\cite{projector-paper}, and a workload from Facebook's datacenter network traces \cite{roy2015inside}. 
We restrict all datasets to $10^6$ requests on: uniform workload with $100$ nodes, HPC with $500$ nodes, ProjectToR with $100$ nodes, and Facebook with $10^4$ nodes.

\vspace{-0.1cm}
\subsection{$k$-ary SplayNet}
\label{sec:karysplaynet}

In this section, we study how $k$-ary SplayNet performs on the chosen workloads. Since one could expect that the cost to physically reconfigure a node with $k$ neighbours ($k$-ary SplayNet) is higher than a node with just three neighbours (standard SplayNet)~--- our main goal is to show that the total routing cost decreases with the increase of $k$. For now, we assume that each of our rotations costs one. This assumption was also made in~\cite{feder2021toward}.

In Tables~\ref{table:hpc}-\ref{table:tc90}, we present how $k$-ary SplayNet works for $k \in [2, 10]$ against the static full and routing-based optimal $k$-ary trees. Note that we implicitly compare our data structures with the standard SplayNet which is the equivalent of our $2$-ary SplayNet. In the first row, we show the total routing cost for $2$-ary SplayNet and the relative cost of other $k$-ary SplayNets (the lower the better). In the second row, we show how our $k$-ary SplayNet works in comparison to the static full $k$-ary tree, i.e., the relative performance of our tree against the static one (the lower the number, the better our tree is).
In the third row, we show how our structure works in comparison to the optimal static routing-based $k$-ary tree (the lower the better).
We compare all workloads except for the uniform one, since we proved in Section~\ref{sec:quasiopt} that full $k$-ary tree is almost optimal. The green color means that our tree performs better, otherwise, worse.

We make two main observations from the tables. Firstly, as expected, the higher the $k$ the lower the total routing cost in $k$-ary SplayNet. Secondly, with the low temporal locality, as for HPC, Facebook, 0.25, and 0.5 workloads, the full $k$-ary tree typically performs better on higher $k$, while the optimal routing-based $k$-ary tree works better but no more than $3$ times. The latter means that our data structure is constant-away from optimality in practice. This behaviour can be simply explained; the higher $k$ the better the cost for $k$-ary full trees on uniform workloads. On high locality, i.e., 0.75 and 0.9, our $k$-ary SplayNet outperforms both static trees.
Also, we need to note that the algorithm building the optimal routing-based tree has high-complexity and we were not able to compute this tree for the Facebook workload.

\noindent
{
\resizebox{\linewidth}{!}{%
    \centering
    \begin{tabular}{|c|c|c|c|c|c|c|c|c|c|}
         \hline
         & 2 & 3 & 4 & 5 & 6 & 7 & 8 & 9 & 10 \\\hline
    SplayNet & 4798648 & 0.87x & 0.82x & 0.75x & 0.76x & 0.73x & 0.70x & 0.69x & 0.70x \\\hline
    Full Tree & \cellcolor{green!25}0.78x & \cellcolor{green!25}0.94x & \cellcolor{red!25}1.04x & \cellcolor{red!25}1.07x & \cellcolor{red!25}1.16x & \cellcolor{red!25}1.17x & \cellcolor{red!25}1.25x & \cellcolor{red!25}1.25x & \cellcolor{red!25}1.29x \\\hline
    Optimal Tree & \cellcolor{red!25}1.52x & \cellcolor{red!25}1.90x & \cellcolor{red!25}2.15x & \cellcolor{red!25}2.22x & \cellcolor{red!25}2.45x & \cellcolor{red!25}2.48x & \cellcolor{red!25}2.49x & \cellcolor{red!25}2.58x & \cellcolor{red!25}2.75x \\\hline
    \end{tabular}
}
    \captionof{table}{The comparison of $k$-ary SplayNet on HPC workload.}
    \label{table:hpc}
}

\noindent
{
\resizebox{\linewidth}{!}{%
    \centering
    \begin{tabular}{|c|c|c|c|c|c|c|c|c|c|}
         \hline
         & 2 & 3 & 4 & 5 & 6 & 7 & 8 & 9 & 10 \\\hline
         SplayNet & 3151626 & 0.93x & 0.91x & 0.87x & 0.84x & 0.86x & 0.86x & 0.84x & 0.83x \\\hline
    Full Tree & \cellcolor{green!25}0.40x & \cellcolor{green!25}0.49x & \cellcolor{green!25}0.46x & \cellcolor{green!25}0.52x & \cellcolor{green!25}0.70x & \cellcolor{green!25}0.50x & \cellcolor{green!25}0.58x & \cellcolor{green!25}0.57x & \cellcolor{green!25}0.92x \\\hline
    Optimal Tree & \cellcolor{red!25}1.45x & \cellcolor{red!25}1.81x & \cellcolor{red!25}2.09x & \cellcolor{red!25}2.10x & \cellcolor{red!25}2.08x & \cellcolor{red!25}2.20x & \cellcolor{red!25}2.22x & \cellcolor{red!25}2.22x & \cellcolor{red!25}2.25x \\\hline
    \end{tabular}
}
    \captionof{table}{The comparison of $k$-ary SplayNet on ProjectToR workload.}
    \label{table:projecttor}
}
\noindent
{
\resizebox{\linewidth}{!}{%
    \centering
    \begin{tabular}{|c|c|c|c|c|c|c|c|c|c|}
         \hline
         & 2 & 3 & 4 & 5 & 6 & 7 & 8 & 9 & 10 \\\hline
         SplayNet & 12320225 & 0.85x & 0.77x & 0.74x & 0.72x & 0.70x & 0.70x & 0.68x & 0.67x \\\hline
        Full Tree & \cellcolor{green!25}0.69x & \cellcolor{green!25}0.87x & \cellcolor{green!25}0.94x & \cellcolor{green!25}1.00x & \cellcolor{red!25}1.07x & \cellcolor{red!25}1.11x & \cellcolor{red!25}1.15x & \cellcolor{red!25}1.19x & \cellcolor{red!25}1.28x \\\hline
        Optimal Tree & \cellcolor{gray!25}- & \cellcolor{gray!25}- & \cellcolor{gray!25}- & \cellcolor{gray!25}- & \cellcolor{gray!25}- & \cellcolor{gray!25}- & \cellcolor{gray!25}- & \cellcolor{gray!25}- & \cellcolor{gray!25}- \\\hline
    \end{tabular}
}
    \captionof{table}{The comparison of $k$-ary SplayNet on Facebook workload.}
    \label{table:facebook}
}
\noindent
{
\resizebox{\linewidth}{!}{%
    \centering
    \begin{tabular}{|c|c|c|c|c|c|c|c|c|c|}
         \hline
         & 2 & 3 & 4 & 5 & 6 & 7 & 8 & 9 & 10 \\\hline
         SplayNet & 1389359 & 0.82x & 0.75x & 0.71x & 0.69x & 0.68x & 0.68x & 0.65x & 0.62x \\\hline
    Full Tree & \cellcolor{green!25}0.99x & \cellcolor{red!25}1.15x & \cellcolor{red!25}1.23x & \cellcolor{red!25}1.30x & \cellcolor{red!25}1.37x & \cellcolor{red!25}1.39x & \cellcolor{red!25}1.47x & \cellcolor{red!25}1.51x & \cellcolor{red!25}1.55x \\\hline
    Optimal Tree & \cellcolor{red!25}1.75x & \cellcolor{red!25}2.12x & \cellcolor{red!25}2.32x & \cellcolor{red!25}2.49x & \cellcolor{red!25}2.64x & \cellcolor{red!25}2.71x & \cellcolor{red!25}2.88x & \cellcolor{red!25}2.99x & \cellcolor{red!25}3.03x \\\hline
    \end{tabular}
}
    \captionof{table}{The comparison of $k$-ary SplayNet on synthetic workload with temporal complexity parameter $0.25$.}
    \label{table:tc25}
}
\noindent
{
\resizebox{\linewidth}{!}{%
    \centering
    \begin{tabular}{|c|c|c|c|c|c|c|c|c|c|}
         \hline
         & 2 & 3 & 4 & 5 & 6 & 7 & 8 & 9 & 10 \\\hline
         SplayNet & 963150 & 0.83x & 0.76x & 0.72x & 0.70x & 0.69x & 0.69x & 0.67x & 0.64x \\\hline
    Full Tree & \cellcolor{green!25}0.69x & \cellcolor{green!25}0.80x & \cellcolor{green!25}0.86x & \cellcolor{green!25}0.91x & \cellcolor{green!25}0.97x & \cellcolor{green!25}0.98x & \cellcolor{red!25}1.03x & \cellcolor{red!25}1.06x & \cellcolor{red!25}1.10x \\\hline
    Optimal Tree & \cellcolor{red!25}1.21x & \cellcolor{red!25}1.49x & \cellcolor{red!25}1.64x & \cellcolor{red!25}1.76x & \cellcolor{red!25}1.87x & \cellcolor{red!25}1.91x & \cellcolor{red!25}2.04x & \cellcolor{red!25}2.12x & \cellcolor{red!25}2.15x \\\hline
    \end{tabular}
}
    \captionof{table}{The comparison of $k$-ary SplayNet on synthetic workload with the temporal complexity parameter $0.5$.}
    \label{table:tc50}
}
\noindent
{
\resizebox{\linewidth}{!}{%
    \centering
    \begin{tabular}{|c|c|c|c|c|c|c|c|c|c|}
         \hline
         & 2 & 3 & 4 & 5 & 6 & 7 & 8 & 9 & 10 \\\hline
         SplayNet & 530049 & 0.85x & 0.78x & 0.75x & 0.73x & 0.72x & 0.72x & 0.70x & 0.67x \\\hline
    Full Tree & \cellcolor{green!25}0.38x & \cellcolor{green!25}0.45x & \cellcolor{green!25}0.49x & \cellcolor{green!25}0.52x & \cellcolor{green!25}0.55x & \cellcolor{green!25}0.56x & \cellcolor{green!25}0.59x & \cellcolor{green!25}0.61x & \cellcolor{green!25}0.64x \\\hline
    Optimal Tree & \cellcolor{green!25}0.68x & \cellcolor{green!25}0.84x & \cellcolor{green!25}0.94x & \cellcolor{red!25}1.02x & \cellcolor{red!25}1.09x & \cellcolor{red!25}1.12x & \cellcolor{red!25}1.19x & \cellcolor{red!25}1.24x & \cellcolor{red!25}1.26x \\\hline
    \end{tabular}
}
    \captionof{table}{The comparison of $k$-ary SplayNet on synthetic workload with the temporal complexity parameter $0.75$.}
    \label{table:tc75}
}
\noindent
{
\resizebox{\linewidth}{!}{%
    \centering
    \begin{tabular}{|c|c|c|c|c|c|c|c|c|c|}
         \hline
         & 2 & 3 & 4 & 5 & 6 & 7 & 8 & 9 & 10 \\\hline
         SplayNet & 271838 & 0.88x & 0.83x & 0.80x & 0.79x & 0.78x & 0.78x & 0.76x & 0.74x \\\hline
    Full Tree & \cellcolor{green!25}0.20x & \cellcolor{green!25}0.24x & \cellcolor{green!25}0.27x & \cellcolor{green!25}0.29x & \cellcolor{green!25}0.31x & \cellcolor{green!25}0.31x & \cellcolor{green!25}0.33x & \cellcolor{green!25}0.34x & \cellcolor{green!25}0.36x \\\hline
    Optimal Tree & \cellcolor{green!25}0.36x & \cellcolor{green!25}0.46x & \cellcolor{green!25}0.53x & \cellcolor{green!25}0.58x & \cellcolor{green!25}0.62x & \cellcolor{green!25}0.64x & \cellcolor{green!25}0.68x & \cellcolor{green!25}0.72x & \cellcolor{green!25}0.73x \\\hline
    \end{tabular}
}
    \captionof{table}{The comparison of $k$-ary SplayNet on synthetic workload with the temporal complexity parameter $0.9$.}
    \label{table:tc90}
}

\subsection{A case study of the centroid heuristic for $k=2$}
\label{sec:3splaynet}

%

In this subsection, we study the online centroid heuristic experimentally for the case of $k=2$. We compare our centroid-based 3-SplayNet (Figure~\ref{fig:3splaynet}) with the $2$-ary SplayNet, which is the standard SplayNet.
We implemented and compared them on different workloads that simulate real request patterns. As a result, it appears that on workloads with low temporal complexity $3$-SplayNet works better than SplayNet.

We run the workloads described above on four different structures: 3-SplayNet, SplayNet, static full binary search tree, and static optimal binary search tree.
In Table~\ref{table:tc90}, one can see the average request cost by $3$-SplayNet and the relative difference with other approaches.
The green cell means that our $3$-SplayNet is better and red means otherwise.

We observe that 3-SplayNet performs better or similarly to SplayNet on average and low temporal complexity workloads (0.25 and 0.5), while on high temporal complexity workloads (0.75 and 0.9) it works a bit worse. 
Also, 3-SplayNet outperforms SplayNet for the uniform, ProjecToR and Facebook workloads, but not for the HPC workload (higher locality than the other two real-world workloads).
We interpret this as the effect of having fixed centroid nodes.

\vspace{0.2cm}
\noindent
{
    \centering
    \begin{tabular}{|c|c|c|c|c|}
         \hline
         & 3-SplayNet & SplayNet & Full Binary Net & Static Optimal Net\\\hline
	Uniform & 17.730 & \cellcolor{green!25} x1.059 & \cellcolor{red!25} x0.789 & \cellcolor{red!25} x0.759\\\hline
	HPC & 9.269 & \cellcolor{red!25} x0.956 & \cellcolor{green!25} x1.206 & \cellcolor{green!25} x1.034\\\hline
	ProjecToR & 2.865 & \cellcolor{green!25} x1.132 & \cellcolor{green!25} x3.040 & \cellcolor{red!25} x0.800\\\hline
	Facebook & 8.210 & \cellcolor{green!25} x1.104 & \cellcolor{red!25} x0.939 & \cellcolor{red!25} x0.852\\\hline
	Temporal 0.25 & 13.332 & \cellcolor{green!25} x1.046 & \cellcolor{green!25} x1.046 & \cellcolor{red!25} x0.937\\\hline
	Temporal 0.5 & 9.414 & \cellcolor{green!25} x1.021 & \cellcolor{green!25} x1.482 & \cellcolor{green!25} x1.326\\\hline
	Temporal 0.75 & 5.520 & \cellcolor{red!25} x0.963 & \cellcolor{green!25} x2.527 & \cellcolor{green!25} x2.250\\\hline
	Temporal 0.9 & 3.186 & \cellcolor{red!25} x0.856 & \cellcolor{green!25} x4.380 & \cellcolor{green!25} x3.862\\\hline
    \end{tabular}
    \captionof{table}{The comparison of $3$-SplayNet with other known nets. Green means that $3$-SplayNet is better.}
    \label{table:tc90}
}

\section{Conclusion and future work}
\label{sec:conclusion}

We presented online and offline algorithms for $k$-ary search tree networks. Specifically, we presented dynamic programming algorithms for computing an optimal static network for generic and uniformly distributed traffic. Then, we presented an online $k$-ary search tree network and a variant that utilizes our centroid structure. We proposed novel splay operations that are applicable in the context of self-adjusting networks. Our experimental results show: 1)~the total routing cost for $k$-ary SplayNet is smaller than for standard SplayNet and is smaller than for full $k$-ary trees on traces with high locality; 2) on binary trees for real and synthetic traces of medium to low locality our centroid network outperforms SplayNet and its performance is always close to the best out of the algorithms tested.
We believe that our work paves the way to new SANs for $k$-ary search tree networks for general and specific traffic patterns.


\newpage
\bibliography{references}
\newpage

\appendix

\section{Dynamic programming algorithms for the optimal tree}
\label{app:dp}

\subsection{Algorithm for arbitrary traffic patterns}
\label{app:genericworkload}
In our first result, we construct an offline (optimal) static routing-based $k$-ary search tree network. We are given the number of nodes $n$ and a demand matrix $D \in \mathbb{N}_{0}^{n \times n}$ that includes the total number of requests between each pair of nodes $u$ and $v$. We have to find a $k$-ary search tree $T$ on $n$ vertices that minimizes the total distance:
	\begin{align*}
		\mathrm{TotalDistance}(D, T) = \sum\limits_{(u, v) \in [n]\times[n]}d_{T}(u, v)\cdot D[u, v],
	\end{align*} 
	where $d_T(u, v)$ is the distance between nodes $u$, $v$ in tree $T$, and $[n] = \{1,\ldots, n\}$.
%
It is convenient to think about this problem in terms of edge potentials.
\begin{definition}
    \label{def:edge_potential}
	Given a demand matrix $D \in \mathbb{N}_0^{n\times n}$ and a tree $T$ on $n$ vertices the potential of an edge $e \in E(T)$ is
		$\mathrm{potential}(D, T, e) = \sum\limits_{(u, v) \in \mathrm{passThrough}(T, e)} D[u, v]$, 
	where $\mathrm{passThrough}(T, e)$ is a set of pairs $(u, v) \in [n]\times[n]$ such that the shortest path connecting $u$ and $v$ in $T$ passes through $e$. Note that the shortest path is unique since $T$ is a tree.
\end{definition}

By their definitions, the total distance can be expressed using the potentials as follows:  $\mathrm{TotalDistance}(D, T) = \sum\limits_{e \in E(T)}\mathrm{potential}(D, T, e).$
This holds because each routing request adds 1 to the potential of each edge of the routing path connecting the communication endpoints.
Now, we present the algorithm.

\begin{theorem}
\label{thm:arbitrary workload}
An offline static routing-based $k$-ary Search Tree network, i.e., one with the minimal total distance given the requests in advance, can be constructed in $O(n^3k)$.
\end{theorem}
\begin{proof}
Our algorithm uses dynamic programming.
Throughout the proof, when we refer to the segment $[i, j]$, we assume that $i \leq j$. 
For a segment $[i, j]$, we denote the submatrix $D[i\ldots j, i\ldots j]$ by $D|_{[i, j]}$.
We define the matrix  $W\in\mathbb{N}_0^{n \times n}$, where
$W[i, j] = \sum\limits_{u \in [n] \setminus [i,j]}\sum\limits_{v \in [i, j]}D[u, v] + D[v, u]$.
Intuitively, the value $W[i, j]$ is the total number of requests going in or out of segment $[i, j]$.
	
\begin{claim}
There is an algorithm that computes $W$ in $O(n^3)$ time.
\end{claim}
\begin{proof}
	We express $W[i, j]$ in terms of forward and backward functions. Namely for each pair of nodes $(u, v),\ u < v$ we define 
		$F[u, v] = \sum\limits_{w = v}^n D[u, w] + D[w, u].$
	That is, it calculates the number of requests between $u$ and $[v, n]$.
	Analogously, we define for each pair of nodes $(u, v),\ v < u$, 
		$B[u, v] = \sum\limits_{w = 1}^v D[u, w] + D[w, u]$.
	The latter calculates the number of requests between $u$ and $[1, v]$.
	The whole prefix function $F[u, \cdot]$ can be computed in $O(n)$. First, we compute $F[u, u + 1]$ by its definition in $O(n)$. Then, we move the left point $v$ of the suffix by one: $F[u, v]$ for $v > u + 1$ is computed as $F[u, v - 1] - D[u, v] - D[v, u]$. Giving us $O(n)$ in total.
	Symmetrically, $B[u, \cdot]$ can be computed in $O(n)$.
	Thus, all matrices $F$ and $B$ can be precomputed in $O(n^2)$, i.e., we iterate over all $u$. 
    Now, we can compute $W[i, j]$ in $O(n)$ using prefix functions $F$ and $B$ in the following manner: $W[i, j] = \sum\limits_{u \in [i, j]} F[u, j + 1] + B[u, i - 1]$, i.e., for each node $u$ in the segment we calculate the number of requests to the left out of the segment and the number of requests to the right out of the segment. Thus, we computed $W$ in $O(n^3)$.
\end{proof}


We define a target $\mathrm{cost}$ of a segment $[i, j]$ as the cost of the optimal routing-based $k$-ary Search Tree built on that segment plus the number of requests going out of that segment calculated in $W$: $\mathrm{cost}(i, j) = \min\limits_{T}\mathrm{TotalDistance}(D|_{[i, j]}, T) + W[i, j]$.
Now, we can define our dynamic programming $dp$ for $1 \leq i \leq j \leq n$ and $1 \leq t \leq k$ as
		$dp[i][j][t] = \min\limits_{i = i_1 < i_2 < \ldots < i_{t + 1} = j + 1}\sum\limits_{p = 1}^t \mathrm{cost}(i_{p}, i_{p + 1} - 1)$.
	
Intuitively, $dp[i][j][t]$ computes the minimal cost of partitioning a segment $[i, j]$ into $t$ children that are $k$-ary search trees. We can compute $dp$ by using the following equalities:
\begin{align*}
        dp[i][j][1] =& \min\limits_{r \in [i, j]}\min\limits_{d_l + d_r \leq k} (dp[i][r - 1][d_l] + dp[r + 1][j][d_r]\\
        & \hspace{2cm}+ W[i, j])\\
        dp[i][j][t] =& \min\limits_{l \in [i, j - 1]} (dp[i][l][1] + dp[l + 1][j][t - 1]),\ t > 1
\end{align*}

The logic is that in order to partition the segment into $t > 1$ trees one should first choose the prefix subsegment for the first tree and build $t - 1$ trees on the remaining segment. 
    
The case $t = 1$ is special. It means that we want to build a single search tree on this segment. In order to do so, we first choose a key $r \in [i,j]$ for the root node \textbf{which also belongs to the routing array} and after that the number of children $d_l$ to the left of a root node $[i, r - 1]$ and a number of children $d_r$ to the right of a root node $[r + 1, j]$. This covers all the possible cases. In each case, we can optimize each subtree out of $d_l + d_r$ subtrees independently, which equals the corresponding $dp$ value. Finally, we add the number of requests that pass an edge from $r$ to the parent, which is $W[i, j]$, i.e., the potential of that edge. Note that the requests to $r$ from the subtrees are already calculated in the corresponding $dp$ using $W$.

When calculating $dp[i][j][t]$ we refer to the answer on subsegments of $[i, j]$, so we make sure that this value is already calculated by processing segments in increasing length. Namely, we start by setting the answer on the segments of length $1$, and, then, proceed by considering all the segments of length $2$, then $3$, and so on, up to $n$. 

We need to consider all possible segments and there are $O(n^2)$ of them. For each segment we calculate $dp[i][j][t]$ for $t \in [1, k]$. When $i, j$ and $t > 1$ are fixed, we spend $O(n)$ time iterating over $[i, j]$ looking for a minimum. This results in $O(n^2 \cdot (k - 1) \cdot n)$ in total. And when $i, j$ and $t = 1$ are fixed, we spend $O(n)$ considering different roots and $O(k^2)$ possibilities of distributing subtrees to the left and to the right of the root. So, in this case we obtain $O(n^3 k^2)$.  

It is possible to reduce the complexity by $k$. For that we introduce $dp_2[i][j][x] = \min\limits_{y \leq x} dp[i][j][y]$. If we can calculate that, then we do not have to iterate over all pairs $d_l$ and $d_r$. Now, for the $t = 1$ case, we need to find $\min\limits_{d_l+d_r=k} dp_2[i][r - 1][d_l] + dp_2[r + 1][j][d_r]$. This gives $O(n^3 k)$ in total. Thus, $dp_2$ can be computed in the desired time.
\end{proof}

\begin{remark}
  We remark that computing an optimal static non-routing-based $k$-ary search tree network is an open problem. That is, there is neither a trivial dynamic programming extension for it nor an NP-hardness proof.
\end{remark}

\subsection{Algorithm for uniformly distributed traffic}
\label{app:uniformopt}
In this section, we improve the cubic complexity proven in Section \ref{app:genericworkload} for the special case of a uniform workload. This case is relevant to all-to-all traffic patterns.
The uniform workload is an infinite workload where each pair of nodes is requested uniformly at random.
Our goal is to find a static $k$-ary search tree that serves an infinite uniform workload as fast as possible, i.e., the expectation of the cost of each query is minimal.
To simplify the analysis we consider a finite version of this workload. Note that a finite and an infinite uniform workload are the same in terms of expected values of query costs.
A finite uniform workload is a workload where each pair of nodes is requested exactly once, thus we are interested in minimizing $\mathrm{TotalDistance}(D_{uniform}, T) = \sum\limits_{(u, v) \in [n]\times[n]}d_{T}(u, v)$, where $D_{uniform}$ is an upper triangular matrix in which all elements in the diagonal and below are $0$ and all remaining elements are $1$.

Note that we are interested in constructing an optimal network which is not necessarily a full static $k$-ary tree for the uniform workload case.
We show how to update the dynamic program of the previous section so that it only requires $O(n^2 \cdot k)$ time.
We remark that the resulting tree is not required to be \emph{routing-based}.

\begin{lemma}
    \label{clm:w uniform}
    In the uniform workload scenario, $W[i, i - 1 + l] = l \cdot (n - l)$ for any $l \in [1, n]$ and any $i \in [1, n - l]$. That is, the values of $W$ for a segment depend only on its length, not its position. 
    %
\end{lemma}
\begin{proof}
    Recall that intuitively $W[i, j]$ indicates the number of requests going out of the segment $[i, j]$. Since each node within the segment communicates exactly once with each node outside the segment, then $W[i, i + l - 1] = l \cdot (n - l)$.
\end{proof}

\begin{lemma}
    \label{clm:cost uniform}
    In the uniform workload scenario, $\mathrm{cost}(i, i  - 1 + l) = \mathrm{cost}(j, j - 1 + l)$ for any $l \in [1, n]$ and any $i,j \in [1, n - l]$. That is, the cost of the segment depends only on its length, not its position. 
\end{lemma}
\begin{proof}
    Recall that $\mathrm{cost}(i, j) = \min\limits_{T}\mathrm{TotalDistance}(D|_{[i, j]}, T)$ $+ W[i, j]$.
    By Lemma \ref{clm:w uniform} the second term is equal for any two segments of an equal length. As for the first term, it is also equal since $D|_{[i, i - 1 + l]} = D|_{[j, j - 1 + l]}$ for all $i, j, l$ in the uniform case. 
\end{proof}

By Lemma \ref{clm:cost uniform}, we can simplify our $dp$ from three parameters $dp[i][j][t]$ to two $dp[l][t]$, where $l$ now signifies the length of the segment. Thus, we can reduce our dynamic program by one dimension and get rid of an $n$ factor, resulting in $O(n^2k)$. 

\section{Centroid static k-ary search tree network in O(n). Full version.}
\label{app:quasiopt}

A $k$-ary search tree can be split on levels: the $i$-th level consists of nodes that are at distance $i + 1$ from the root.
The tree is \emph{weakly-complete} when all its levels, except for the last one, are fully filled (i.e., the $i$-th level has $k^{i-1}$ nodes).
Nodes on the last level can be distributed arbitrarily.
%
The \emph{height} of a tree is the length of the path in edges from the root to the nodes on the last non-empty level.
%
Moreover, in a finite uniform workload, a potential from Definition~\ref{def:edge_potential} of an edge connecting two subtrees $S$ and $T$ is equal to $|V(S)| \cdot |V(T)|$. 

Our goal is to find an optimal non-rooted $(k+1)$-degree tree instead of a $k$-ary search tree. A $(k+1)$-degree tree is a non-rooted tree where each node has at most $k+1$ neighbours. Such trees represent the same set of trees as $k$-ary search trees: you can root a $(k+1)$-degree tree by a leaf and obtain a $k$-ary search tree. Also, we can ignore the search property for now, since we consider just the uniform workload; given a rooted tree we can always place the right labels. The main intuition of our improvement is that when someone talks about $k$-ary tree they consider the root to have $k$ children, while, in general, it can have $k+1$ children.

\begin{definition}
A \emph{centroid} $(k+1)$-degree tree is a tree with the root having $k + 1$ weakly-complete $k$-ary trees. All the levels of the whole tree are fully filled except possibly the last one. We can change the relative positions of subtrees such that the leaves on the last level are all grouped together to the left. The tree is shown on Figure~\ref{fig:quasi-optimal tree}.
\end{definition}

\begin{figure}[H]
    \centering
    \includegraphics[scale=0.3]{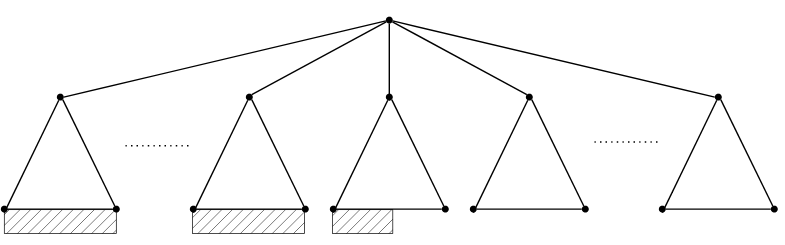}
    \caption{A centroid tree after the reposition of subtrees. Hatched rectangles represent leaves.}
    \label{fig:quasi-optimal tree}
\end{figure}




Now, we prove that the centroid $(k+1)$-degree tree has total distance close to the total distance of the optimal $(k+1)$-degree tree.

\begin{definition}
Consider two neighbouring weakly-complete subtrees. A \emph{push-up} operation moves a leaf from the last level of one tree to the last level of another. 
\end{definition}

\begin{lemma}
	\label{lem:push_up}
	Assume that we do a push-up operation in the tree $G$ from the weakly-complete subtree $T$ of height $h_2$ to a weakly-complete sibling subtree $S$ of height $h_1$ ($h_2 > h_1$, note that the tree is not necessarily weakly-complete, i.e., $h_2$ can be bigger than $h_1 + 1$).
Assuming $|V(T)| + |V(S)| \leq \frac{n}{9k}$, the total distance for uniform workload decreases.
\end{lemma}
\begin{proof}
	Let a leaf $u$ of $T$ be the removed node and let a leaf $v$ of $S$ is where we place the moved node. The total distance is affected: all the terms with $u$, i.e., $d_G(u, x)$ terms, are removed and $n - 1$ new terms $d_G(v, x)$ with $v$ are added, where $G$ is the whole tree.
	Let us denote the path from the root of $T$ to $u$: $\mathrm{root} = t_1, t_2, \ldots, t_{h_2 + 1} = u$. Let us also denote the subtree formed by $t_i$ and its child trees other than the one with $u$ as $T_i$.
	The same notation we use for the tree $S$ and the node $v$. You can see the tree at Figure~\ref{fig:pushup}.
	We define $R := G \setminus (S \cup T)$.
	\begin{figure}[H]
		\centering
		\includegraphics[scale=0.3]{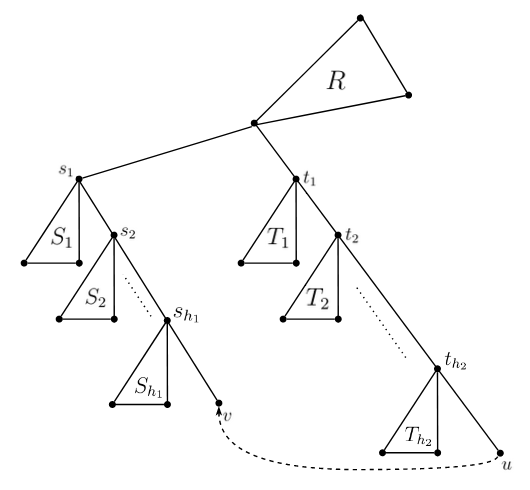}
		\caption{Push up Lemma notation}
		\label{fig:pushup}
	\end{figure}
		Our goal now is to calculate the difference in the total distance after the push up operation. This difference consists of three parts:
	\begin{enumerate}
		\item 
		We denote the change in the distance for nodes in $S$ ($v$ is closer to them than $u$):
		\begin{align*}
			\Delta_1 := \sum\limits_{w \in V(S)}d_{G}(u, w) - \sum\limits_{w \in V(S)}d_{G}(v, w).
		\end{align*} 
		\item 
		We denote the change in the distance for nodes in $T$ ($u$ is closer to some nodes than $v$, note that not all nodes in $T$ are closer to $u$ since $v$ has smaller depth):
		\begin{align*}
			\Delta_2 := \sum\limits_{w \in V(T)}d_{G}(v, w) - \sum\limits_{w \in V(T)}d_{G}(u, w).
		\end{align*} 
		\item 
		We denote the change in the distance for nodes in $R$ ($v$ is closer to them than $u$ since the depth of $v$ is smaller):
		\begin{align*}
			\Delta_3 := \sum\limits_{w \in V(R)}d_{G}(u, w) - \sum\limits_{w \in V(R)}d_{G}(v, w).
		\end{align*} 
	\end{enumerate}
	
	Hence, the total distance is changed by $\Delta := \Delta_2 - \Delta_1 - \Delta_3$. We want to prove that $\Delta$ is negative, thus, the total distance decreases when we move $u$ to $v$.
	
	We can lower bound $\Delta_3$ easily: 
	$
		\Delta_3 = (h_2 - h_1) \cdot |V(R)| 
		\geq (h_2 - h_1)\cdot (9k - 1)(|V(T)| + |V(S)|) > \\ > 8k|V(T)|
	$
	
	Let us now find an upper bound for $\Delta_2$.
	At first, we define the Total Distance to Root (or TDR) function for tree $W$ rooted at $r$ as $TDR(W) = \sum\limits_{v \in V(W)} d_{W}(v, r)$.
    We say that $T_i$ is rooted at $t_i$.

    So, the total distance from $u$ and $v$ to the nodes of $T$ can be expressed in terms of this function TDR:
	\begin{align*}
		\sum\limits_{w \in V(T)}d_{G}(v, w) = \sum\limits_{i = 1}^{h_2}\left(|V(T_i)| \cdot (h_1 + i + 1) + TDR(T_i)\right)\\
		\sum\limits_{w \in V(T)}d_{G}(u, w) = \sum\limits_{i = 1}^{h_2}\left(|V(T_i)| \cdot (h_2 - i + 1) + TDR(T_i)\right)
	\end{align*}
	
	The intuition behind those formulas is that in order to travel from $v$ to all the nodes in $T_i$ we first need to travel to its root $t_i$ which is at the distance $(h_1 + i + 1)$. We do it for each node, so $|V(T_i)|$ times. And, then, we travel from the root of $T_i$ to a corresponding node, accumulating $TDR(T_i)$ in total. The same is calculated for $u$ but the distance to $t_i$ decreases. 
	
	Thus,
\begin{footnotesize}
	\begin{align*}
		\Delta_2 &= \sum\limits_{i = 1}^{h_2}|V(T_i)|(2i + h_1 - h_2)\\
		&=-(2(h_2 - 1) + h_1 - h_2) +  \sum\limits_{i = 1}^{h_2 + 1}|V(T_i)|(2i + h_1 - h_2)\\
		&= -(h_2 + h_1 - 2) + (h_1- h_2)\sum\limits_{i=1}^{h_2 + 1}|V(T_i)| + 2\sum\limits_{i =1}^{h_2 + 1}i\cdot |V(T_i)|\\
		&= -(h_2 + h_1 - 2) + |V(T)|(h_1 - h_2) + 2\sum\limits_{i =1}^{h_2 + 1}i\cdot |V(T_i)|\\
		& \leq 1 - |V(T)| + 2\sum\limits_{i =1}^{h_2 + 1}i\cdot |V(T_i)| \leq 2\sum\limits_{i =1}^{h_2 + 1}i\cdot |V(T_i)|
	\end{align*}
\end{footnotesize}

	Each $T_i$ consists of root $t_i$ and not more than $k - 1$ weakly-complete $k$-ary trees of height not exceeding $h_2 - i$. Therefore,
	
	\begin{align*}
		|V(T_i)| &\leq 1 + (k - 1) \cdot (1 + k + k^2 + \ldots + k^{h_2 - i})\\
		&= 1 + (k - 1) \cdot \frac{k^{h_2 - i + 1} - 1}{k - 1}\\
		&= k^{h_2 - i + 1}
	\end{align*}

    By that, we notice that 
    \begin{footnotesize}
	\begin{align*}
	    k|V(T)| \geq k(1 + k + \ldots + k^{h_2 - 1}) = \frac{k^{h_2 + 1} - 1}{k - 1} - 1 > \frac{k^{h_2 + 1} - 1}{k - 1},
	\end{align*}
    \end{footnotesize}

    Using this inequality, we can upper bound $\Delta_2$ further.

	\begin{claim}
		\label{clm:push_up}
		$\Delta_2 \leq 8k|V(T)|$
	\end{claim}
	\begin{proof}
		$
			\Delta_2 \leq 2k^{h_2 + 1}\sum\limits_{i =1}^{h_2 + 1}\frac{i}{k^i}
			\leq  2k^{h_2 + 1}\frac{k}{(k-1)^2}
	 		= \\ = 2\frac{k^{h_2 + 1} - 1}{k - 1} \cdot \frac{k}{k - 1} + \frac{2k}{(k-1)^2}
	 		\leq \frac{2k}{k - 1}k|V(T)| + 4
	 		\leq 8k|V(T)|
		$
	\end{proof}
	
	Recall that $\Delta_3 > 8k|V(T)|$ and $\Delta_1 \geq 0$. So we obtain that $\Delta < 0$.
\end{proof}

\begin{corollary}
	\label{cor:push_up1}
	For each subtree $S$ of an optimal tree $T$ if $|V(S)| \leq \frac{|V(T)|}{9k}$ then $S$ is a weakly-complete tree.
\end{corollary}
\begin{proof}
	We prove this statement by the induction on the height of $S$.
	
	At first, we prove the base. If $S$ is of height $2$, then all of its subtrees are either empty or of height $1$ or $0$ and, thus, they are weakly-complete. Suppose that $S$ is not weakly-complete. Thus, we deduce that there is a subtree of height $1$ and an empty subtree, so we can perform a push up operation improving the total distance which contradicts the optimality of $T$.
	
	So, now, we may assume by the inductive hypothesis that all the subtrees of $S$ are weakly-complete.
	
	Suppose that $S$ is not weakly-complete. By the induction hypothesis, there are two subtrees of $S$ such that we can perform a push-up operation between them decreasing the total cost which contradicts the optimality of $T$.
\end{proof}

The next corollary is proven using the Claim~\ref{clm:push_up}.
\begin{corollary}
	\label{cor:push_up2}
	Assume that we do a push-up operation from the weakly-complete subtree $T$ of height $h_2$ to a weakly-complete sibling subtree $S$ of height $h_1$ ($h_1 < h_2$, where $h_1$ is calculated after moving the node and $h_2$ is calculated before moving the node). The total distance increases by $O(kn)$.
\end{corollary}

\begin{definition}
    A \emph{centroid} of a tree $T$ is a node $c \in V(T)$ such that when removed, $T$ will be split into $m$ subtrees $T_1, \ldots, T_m$ with $|V(T_i)| \leq \frac{|V(T)|}{2}$ for all $i$.
    The \emph{centroid decomposition} is represented by $\{c\} \cup \{T_1, \ldots, T_m\}$. 
\end{definition}

\begin{claim}[Jordan~\cite{jordan1869assemblages}]
Any tree has a centroid decomposition.
\end{claim}

Assume we know the optimal $(k+1)$-degree tree $T$ of size $n$ with the total distance. We make a centroid decomposition of it obtaining a centroid $C$ and $k + 1$ trees $T_1, T_2, \ldots T_{k+1}$. ($T_i$ could be empty.)
From now on we assume that $T$ is rooted at $C$.

\begin{figure}[H]
	\centering
	\includegraphics[scale=0.3]{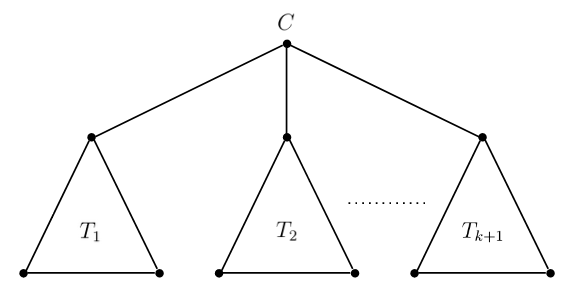}
	\caption{Decomposition of $T$}
\end{figure}

\begin{lemma}
\label{lem:subcentroid}
If we root an optimal tree $T$ at its centroid, then, for each subtree $S = T_i$ its centroid is either a root of $S$ or a child of a root.
\end{lemma}
\begin{proof}
	Let $S$ be a subtree $T_i$ of the optimal tree $T$. Denote its root as $r$.
	
	If all subtrees of $r$ have size $\leq \frac{|V(S)|}{2}$, then $r$ is a centroid of $S$ and the statement holds.
	
	Denote the subtrees of $r$ as $S_{1}, \ldots, S_{k}$.
	
	If $r$ is not a centroid, one of $S_{j}$ is bigger than $\frac{|V(S)|}{2}$. Suppose, for simplicity, it is $S_{1}$. We now prove that the root of $S_{1}$ is a centroid of $S$.
	
	Denote the root of $S_{1}$ as $r_{1}$ and its subtrees as $S_{11}, \ldots, S_{1k}$. The visualisation for the lemma is presented in Figure \ref{fig:centroid decomposition}.
	
	Assume $r_{1}$ is not a centroid. Then, it means that either $\left(\bigcup\limits_{i \in [2, \ldots, k]} S_{i} \right) \cup \{r\}$ is bigger than $\frac{|V(S)|}{2}$ or $S_{1j}$  is bigger than $\frac{|V(S)|}{2}$ for some $j \in [1,\ldots,k]$.
	
	\begin{itemize}
		\item Suppose that $\left(\bigcup\limits_{i \in [2, \ldots, k]} S_{i}\right) \cup \{r\}$ is bigger than $\frac{|V(S)|}{2}$. 
		This is impossible since we already know that $|V(S_{1})| > \frac{|V(S)|}{2}$.
		
		\item $S_{1i}$  is bigger than $\frac{|V(S)|}{2}$ for some $i \in [1,\ldots,k]$. Suppose, for simplicity, this tree is $S_{11}$.
		
		Now, we are going to prove that if we swap $S_{11}$ with any of $S_{i}$, say $S_{2}$ to be certain, the total cost will decrease.
		
		We refer to the total cost expressed in terms of edge potentials. Note that the potential for the edges within $S_i, S_{ij}$ and $R$ does not change. Neither does it change for the edges going out of $S_i$ and $S_{ij}$. So, the only change is in the potential of $(r_{1}, r)$ edge. 
		
		The old value for its potential is
		\begin{align*}
			\left(1 + |V(S_{11})| + \sum\limits_{i = 2}^k |V(S_{1i})|\right) \cdot \\ \cdot \left(|V(R)| + |V(S_{2})| + \sum\limits_{i = 3}^k|V(S_{j})|\right),
		\end{align*}
		
		while the new one is 
		\begin{align*}
			\left(1 + |V(S_{2})| + \sum\limits_{i = 2}^k |V(S_{1i})|\right) \cdot \\ \cdot \left(|V(R)| + |V(S_{11})| + \sum\limits_{i = 3}^k|V(S_{i})|\right).
		\end{align*}
		
		Now we calculate the difference between the potentials:
		\begin{align*}
			\left(|V(S_{11})| - |V(S_{2})|\right)\cdot \left(|V(R)| + \sum\limits_{i = 3}^k|V(S_{i})|\right. - \\
   \left.\sum\limits_{i = 2}^k |V(S_{1i})| - 1\right),
		\end{align*}
		
		which is positive since:
		1)~$|V(S_{11})| > |V(S_{2})|$; and 2)~due to the fact that $C$ is a centroid we know that $|V(R)|$ is bigger than the half of the tree, or in other words:
		$|V(R)| \geq \frac{n}{2} > \sum\limits_{i = 2}^k |V(S_{1i})| + 1$
	\end{itemize}
	\begin{figure}[H]
		\centering
		\includegraphics[scale=0.3]{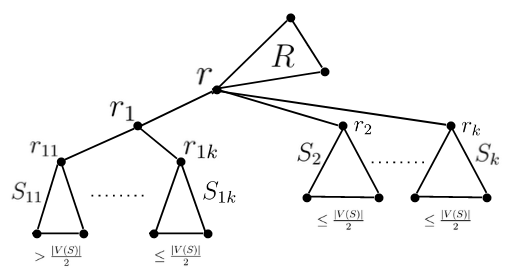}
		\caption{Notation for Lemma \ref{lem:subcentroid}}
		\label{fig:centroid decomposition}
	\end{figure} 
\end{proof}

So we can only have two possibilities for the inner structure of each subtree of an optimal tree (rooted at its centroid).

\begin{figure}[H]
	\centering
	\begin{minipage}[b]{0.35\textwidth}
		\includegraphics[scale=0.4]{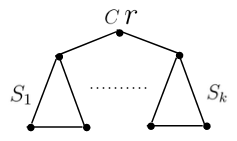}
		\caption{Case 1}
		\label{fig:local_decomposition_v1}
	\end{minipage}
	\hfill
	\begin{minipage}[b]{0.6\textwidth}
		\includegraphics[scale=0.4]{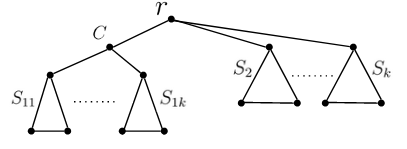}
		\caption{Case 2}
		\label{fig:local_decomposition_v2}
	\end{minipage}
\end{figure}

\begin{corollary}
	\label{cor:subcentroid}
	If we root an optimal $(k+1)$-degree tree $T$ at its centroid, then $|V(S)| \leq \frac{|V(T)|}{2^h}$ holds for each subtree $S$ at level $2h$. 
\end{corollary}
\begin{proof}
	We prove this statement by the induction.
	
	The statement holds for $h = 0$.
	
	Consider a subtree $Q$ at level $2 \cdot (h + 1)$ and a tree $P$ rooted at a grandparent of a root of $Q$. By Lemma \ref{lem:subcentroid}, $|V(Q)| \leq \frac{|V(P)|}{2}$ which by induction hypothesis $\leq \frac{|V(T)|}{2^h \cdot 2} = \frac{|V(T)|}{2^{h + 1}}$
\end{proof}

Combining Corollary~\ref{cor:push_up1} and Corollary~\ref{cor:subcentroid}, we obtain that each subtree of an optimal tree at level $\geq 2\lceil \log_2(9k) \rceil$ must be weakly-complete.

The following lemma is also proved straightforwardly. We suggest the opposite and try to move nodes in between subtrees. We calculate the difference and show that the total cost decreases.
\begin{lemma}
    \label{lem:leaves_movement}
    If there are two neighbouring subtrees of the same height both having their last level not empty and not full, the total distance can be decreased.
\end{lemma}
\begin{proof}
    We consider two such subtrees with the smallest height, in a sense that they have all their leaves as far to the ``left'' as possible. 
    
    Suppose that in each subtree the last level contains at most $m$ leaves. The left subtree has $0 < l < m$ leaves on its last layer and the right subtree has $0 < r < m$ leaves on its last layer. Furthermore, we can assume that $r \leq l$. (The other way around is symmetrical)
    
    We consider two cases: either $l + r \leq m$ (Fig. \ref{fig:leaves_movement1}) or $l + r > m$ (Fig. \ref{fig:leaves_movement2}). 
    \begin{figure}[H]
    	\centering
    	\begin{minipage}[b]{0.45\textwidth}
    		\includegraphics[scale=0.2]{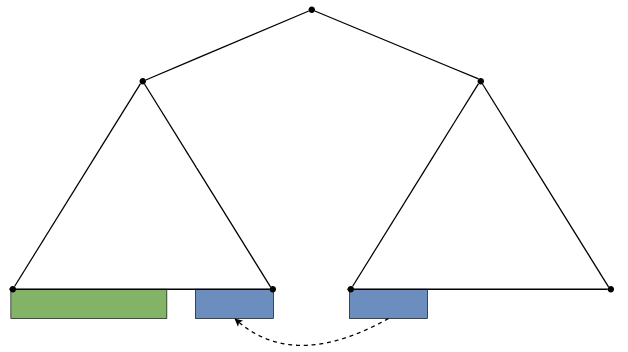}
    		\caption{Case 1}
    		\label{fig:leaves_movement1}
    	\end{minipage}
    	\hfill
    	\begin{minipage}[b]{0.45\textwidth}
    		\includegraphics[scale=0.2]{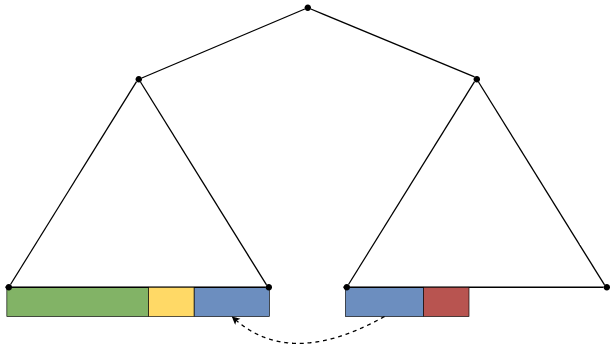}
    		\caption{Case 2}
    		\label{fig:leaves_movement2}
    	\end{minipage}
    \end{figure}
    \begin{itemize}
        \item In the first case, $l + r \leq m$. Consider Figure~\ref{fig:leaves_movement1}.
        Leaves of the left tree are depicted green, leaves of the right tree are depicted blue.
        In this case we move all the leaves from the right tree to the right-most positions in the left tree.
        \begin{itemize}
            \item The total distance among blue leaves did not change.
            \item The total distance between blue leaves and the nodes outside considered trees did not change.
            \item The total distance between blue nodes and the right subtree (without leaves) is now equal to the total distance between blue nodes and the left subtree (without leaves). And vice versa. 
            \item The total distance between blue nodes and green nodes is decreased.
        \end{itemize}
        
        \item In the second case, $l + r > m$. Consider Figure~\ref{fig:leaves_movement2}.
        We move $m - l$ left-most leaves from the right tree to the left tree. 
        Left-most $m-r$ leaves of the left tree are depicted green. Other leaves of the left tree are depicted yellow. Moved leaves of the right tree are depicted blue. Other leaves of the right tree ($r - (m - l)$ of them) are depicted red.
        \begin{itemize}
            \item The total distance among blue leaves did not change.
            \item The total distance between blue leaves and the nodes outside considered trees did not change.
            \item The total distance between blue nodes and the right subtree (without leaves) is now equal to the total distance between blue nodes and the left subtree (without leaves). And vice versa. 
            \item The total distance between blue nodes and green nodes is decreased.
            \item The total distance between blue nodes and red nodes is now equal to the total distance between blue nodes an yellow nodes. And vice versa.
        \end{itemize}
    \end{itemize}
\end{proof}

\begin{lemma}
If a tree $T$ has all its levels filled except possibly the last one, and there are no two neighbouring subtrees of the same height both having their last level not empty and not full, then, we can change the order on the children of each node of $T$, so, that all the leaves of the last level are placed as left as possible, one by one.
\end{lemma}
\begin{proof}
    We prove this statement by an induction on the height.
    If the tree has height one, we simply can make the ``leftmost'' numbering on the non-empty children.
    Now, we discuss the case when the tree has height $h > 1$.
    By the statement of the lemma, there can be at most one child of a root with its last level not fully filled and not empty. By the induction hypothesis, we assume that its leaves on the last level are placed as ``left'' as possible. We now order the sons of the root in the following manner. First, we add the subtrees which have their last level full from left to right. Note, that they can be in arbitrary order. Then, we place the subtree with the last level not full and not empty (it might not exists, but if it exists, there is only one such subtree). By the induction statement, all its leaves are at the left. Finally, we place the subtrees which have their last level empty, in arbitrary order. We got exactly the ``leftmost'' position of leaves. 
\end{proof}

\begin{theorem}
	\label{thm:flatten}
	The difference in the total distance between an optimal $(k+1)$-degree tree $T$ and our centroid $(k+1)$-degree tree is $O(n^2k\log k)$.
\end{theorem}
\begin{proof}
	Our plan is to reconfigure $T$ into the centroid tree while controlling the increase of the total distance.
	
	We push-up some nodes to ensure that all the subtrees at level $l$ are weakly-complete starting from $l = 2\lceil \log_2(9k) \rceil$ and up to $l = 0$ (the whole tree).
	
	Suppose we want to make a subtree $S$ at level $l$ weakly-complete. Since we go through levels decreasingly, we can argue that all the subtrees of $S$ are already weakly-complete (this holds for $l = 2\lceil \log_2(9k) \rceil$).
	
	If $S$ is not weakly-complete, it means that there are two subtrees of $S$ with height difference at least $2$. So, we take the subtree with the biggest height and the subtree with the smallest height and perform a push-up operation between them.
	
	We act in this manner while there are two subtrees of height difference $\geq 2$. Once there are none, we say that $S$ is weakly-complete by definition.
	
	When processing a certain level, each node is moved at most once (in its subtree), so each node is moved no more then $2\lceil \log_2(9k) \rceil$ times, thus, by Corollary~\ref{cor:push_up2}, the total cost change is $O(n^2k\log k)$: at most $n$ nodes move $O(\log k)$ times each increases by $O(nk)$.
	
	And the last step would be to reshuffle leaves on the last level, so they are as far left as possible, so we get a centroid tree.
	
	Assume there are two neighbouring subtrees such that their last level is not empty and not full. By Lemma \ref{lem:leaves_movement}, we can decrease the cost by moving leaves from the smaller subtree to the bigger one. 
	
    We perform those movements until our tree becomes the centroid one. 
\end{proof}

\begin{theorem}
	Assuming $k$ is a constant, the total distance in the optimal $(k+1)$-degree tree $T$ is $\Omega(n^2\log n)$.
\end{theorem}
\begin{proof}
	We root tree $T$ by its centroid $C$. At least two subtrees of $T$, $T_i$ and $T_j$, have at least $\frac{n}{2k}$ nodes. Otherwise, $C$ is not the centroid. Each such $T_i$ has $\Omega(n)$ nodes at levels $\geq \log_k n - 2$. Thus, the total pairwise distance between these nodes is $\Omega(n^2\log n)$.
\end{proof}

\begin{remark}
We can get $k$-ary search tree out of $(k+1)$-degree centroid tree by rooting at some leaf and setting the labels correspondingly. We name such a tree~--- a \emph{centroid $k$-ary search tree}.

Since we consider the uniform workload we know that our centroid tree has the total cost of requests close to the optimal, i.e., misses by at most $O(n^2)$ while the total optimal cost is $\Omega(n^2 \log n)$. Thus, our centroid $k$-ary search tree has
an approximation ratio $1+O(\frac{1}{\log n})$.

\end{remark}

\begin{theorem}
The centroid $k$-ary search tree can be built in $O(n)$.
\end{theorem}
\begin{proof}
At first, we simply build the $(k+1)$-degree centroid tree (we can do this recursively, since we know the sizes of all subtrees). Then, we root it by some leaf and calculate the sizes of all subtrees of this rooted tree. Finally, we need to make this tree to satisfy search property~--- for that we go recursively from the chosen root and set the labels using the precomputed sizes of subtrees. All these stages work in $O(n)$.
\end{proof}
Finally, we show that the full $k$-ary tree also has total distance close to the cost of the optimal tree.

\begin{lemma}
    The total distance in the full $k$-ary tree and the total distance in the centroid $(k+1)$-degree tree are both $n^2\log_kn + O(n^2)$. In other words, they are close in the total cost to the optimal one by $O(n^2)$.
\end{lemma}
\begin{proof}
    Recall that the total distance can be calculated as $\sum\limits_{e \in E} |V(T_e^1)| \cdot |V(T_e^2)|$, where $T_e^1$ and $T_e^2$ are trees that are left when $e$ is removed.

    We first do the calculation for the full $k$-ary tree. Let $l$ be the number of edge levels in the $k$-ary tree on $n$ vertices. There are $k^i$ edges on $i$-th level for $i < l$. We might not care about the edges on the last level, because they contribute at most $n^2$ to the total sum. So the total sum in $k$-ary tree is
\begin{footnotesize}
    \begin{align*}
        \sum\limits_{i = 1}^{l - 1}k^i \cdot \frac{n}{k^i}\cdot\left(n - \frac{n}{k^i}\right) + O(n^2) = n^2 \cdot \sum\limits_{i = 1}^{l -1}\left(1 - \frac{1}{k^i}\right) + O(n^2) = \\ = n^2l + O(n^2) = n^2\log_kn + O(n^2)
    \end{align*}
\end{footnotesize}

    Now let's analyze $(k+1)$-degree tree. We follow the same logic but with $l'$ being the number of levels in the tree.
\begin{footnotesize}
    \begin{align*}
        (k+1) \cdot \frac{n}{k + 1} \cdot (n - \frac{n}{k + 1}) + \\ + \sum\limits_{i = 1}^{l' - 2}(k + 1)k^i \cdot \frac{n}{(k + 1) k^i}\cdot(n - \frac{n}{(k + 1) \cdot k^i}) + O(n^2) =\\ n^2 \cdot \left(\frac{1}{k + 1} + \sum\limits_{i = 1}^{l' - 2}\left(1 - \frac{1}{(k + 1) \cdot k^i}\right)\right) + O(n^2) = n^2l' + O(n^2) = \\ = n^2\log_kn + O(n^2).
    \end{align*}
\end{footnotesize}
    The last equality is since $l' = 1 + \log_k(\frac{n}{k+1}) = \log_kn + O(1)$
\end{proof}

\begin{remark}
\label{rem:expOptCentroid}
The results of the last Lemma show that the full and centroid trees are close to the optimal. However, in the uniform workload the centroid tree should have better total cost, since we split in the centroid vertex by $k+1$ balanced subtrees.
In our experiments, we found that our \emph{centroid} $k$-ary search tree \textbf{is indeed optimal} for all $n$ less than $10^3$ when $k$ is up to $10$.
\end{remark}

\end{document}